\documentclass{article}
\usepackage[utf8]{inputenc}
\usepackage[left=2cm,right=2cm,top=2cm,bottom=2cm]{geometry}
\usepackage{graphicx,amsfonts,amsmath,amssymb,stmaryrd} 
\usepackage{xcolor}
\usepackage{appendix}
\usepackage[all]{xy}
\usepackage{tikz-cd}
\usetikzlibrary{decorations.pathmorphing} 
\usepackage{quiver}
\usepackage{url}
\usepackage{hyperref}
\usepackage[hyphenbreaks]{breakurl}
\usepackage{polski}
\usepackage[english]{babel}

\newcommand{\ds}{\displaystyle}

\newcommand{\R}{\mathbb{R}}
\newcommand{\C}{\mathbb{C}}
\newcommand{\K}{\mathbb{K}}

\newcommand{\OConf}{\mathrm{Conf}}
\newcommand{\UConf}{\mathrm{C}}

\newcommand{\VB}{\mathcal{VB}}
\newcommand{\bA}{\mathbf{A}}

\newcommand{\bC}{\mathbf{C}}

\newcommand{\bV}{\mathbf{V}}
\newcommand{\bW}{\mathbf{W}}

\newcommand{\bI}{\mathbf{I}}
\newcommand{\bP}{\mathbf{P}}
\newcommand{\bT}{\mathbf{T}}
\newcommand{\bS}{\mathbf{S}}
\newcommand{\bL}{\mathbf{\Lambda}}
\newcommand{\bSigma}{\mathbf{\Sigma}}

\newcommand{\D}{\mathcal{D}}

\renewcommand{\L}{\mathcal{L}}

\newcommand{\F}{\mathrm{F}}
\newcommand{\GL}{\mathrm{GL}}

\newcommand{\id}{\mathrm{id}}

\newcommand{\Hom}{\mathrm{Hom}}
\newcommand{\Mor}{\mathrm{Mor}}

\newcommand{\USh}{\mathrm{USh}}
\newcommand{\Dens}{\mathrm{Dens}}
\usepackage[size=tiny, backgroundcolor=white]{todonotes}
\newcommand{\sgn}{\mathrm{sgn}}

\newcommand{\exttens}{{\boxtimes^{\mbox{\tiny{ext}}}}}

\makeatletter
\newcommand{\boxedsymbol}[2]{%
 \begingroup
 \setlength{\fboxsep}{0pt}%
 \fbox{%
 $\m@th#1\mspace{-1.25mu}#2\mspace{-1.25mu}$%
 }%
 \endgroup
}
\makeatother
\newcommand{\boxwedge}{\mathbin{\mathpalette\boxedsymbol\wedge}}

\newtheorem{theorem}{Theorem}[subsection]
\newtheorem{proposition}[theorem]{Proposition}
\newtheorem{corollary}[theorem]{Corollary}
\newtheorem{lemma}[theorem]{Lemma}
\newtheorem{defin}[theorem]{Definition}
\newenvironment{definition}{\begin{defin} \em}{\end{defin}\par\vspace{.2cm}}
\newtheorem{rem}[theorem]{Remark}
\newenvironment{remark}{\begin{rem} \em}{\hfill$\Box$\end{rem}\par\vspace{.2cm}}
\newtheorem{exa}[theorem]{Example}
\newenvironment{example}{\begin{exa} \em}{\end{exa}\par\vspace{.2cm}}
\newenvironment{proof}{\bigskip \noindent{\bf Proof.\/}} {\hfill$\Box$\par\vspace{.2cm}}

\begin{document}

\title{Poisson bundles over unordered configurations}

\author{Alessandra Frabetti\thanks{Université Lyon 1, Centrale Lyon, INSA Lyon, Université Jean Monnet, CNRS, ICJ UMR5208, 69622 Villeurbanne, France.} ~ and Olga Kravchenko\footnotemark[1] ~and Leonid Ryvkin\footnotemark[1] \thanks{University of Göttingen}}

\date{\today}
\maketitle
\begin{center}
 \textit{In memory of Krzysztof Gaw\k{e}dzki}
\end{center}
\begin{abstract}
In this paper we construct a Poisson algebra bundle whose distributional sections are suitable to represent multilocal observables in classical field theory. 
To do this, we work with vector bundles over the unordered configuration space of a manifold $M$ and consider the structure of a $2$-monoidal category given by the usual (Hadamard) tensor product of bundles and a new (Cauchy) tensor product which provides a symmetrized version of the usual external tensor product of vector bundles on $M$. 
We use the symmetric algebras with respect to both products to obtain a Poisson 2-algebra bundle mimicking the construction of Peierls bracket from the causal propagator in field theory. 
The explicit description of observables from this Poisson algebra bundle will be carried out in a forthcoming paper. 
\end{abstract}

\tableofcontents


\subsection*{Acknowledgements}

The authors would like to thank Stanislas Herscovich, Antonio Miti, Hai Chau Nguyen and Volodya Roubtsov for fruitful discussions related to the project. The authors have worked on the present manuscript during their stay at IHP in the summer of 2023 and thank the grant Accueil EC 2023-2024 of the Conseil Académique de l’UCBL and the LabEx MiLyon. L.R. acknowledges the support of the DFG through the grant Higher Lie Theory. A.F. thanks the RTG2491 for supporting multiple visits to the Mathematical Institute in Göttingen. For the purpose of Open Access, a CC-BY-NC-SA public copyright license has been applied by the authors to the present document and will be applied to all subsequent versions up to the Author Accepted Manuscript arising from this submission.

\section{Introduction}
\label{sec:intro}

In relativistic field theory, it is a longstanding open problem to describe the algebra of observables in a fully covariant way. In classical field theory, this problem motivated the introduction of multisymplectic geometry \cite{Gotay-Isenberg-Marsden-Montgomery-1998,Forger-Romero-2005,Blohmann-2023}; in quantum field theory, it gave rise to perturbative Algebraic Quantum Field Theory \cite{Brunetti-Fredenhagen-Verch-2003, Rejzner-2016}. One of the key open questions is the covariant description of multilocal observables, which are products of local observables. 

For a field theory given by a vector bundle $E\to M$ over a spacetime manifold $M$, local observables can be described by (distributional) sections of the dual vector bundle $(JE)^*\otimes \Dens_M$ of the jet bundle (of a given order) with values in densities. Multilocal observables are usually related to integrals of sections of the repeated external product of $(JE)^*\otimes \Dens_M$ over an arbitrary number of copies of $M$, but this representation is not consistent, since the product of observables is commutative while the external tensor product of vector bundles is not. 

In this article and its companion \cite{Frabetti-Kravchenko-Ryvkin-2024}, we propose a geometric description of multilocal observables in terms of distributional sections of vector bundles over a comprehensive base manifold compatible with the switch of points in $M$. 
Hence in this first part of the program we develop the algebraic / geometric background, whereas in the second part we focus on functional-analytic considerations and show the application in field theory.\\

The key algebraic / geometric ingredients to get a covariant Poisson algebra of observables are:
\begin{itemize}
\item The space $\UConf(M)$ of unordered configurations of distinct points of $M$, equipped with the structure of a manifold (of nonpure dimension). 
Configuration spaces appear in many branches of mathematics since the 1960s \cite{Fox-Neuwirth-1962}, especially in topology \cite{Cohen-Cohen-Kuhn-Neisendorfer-1983} \cite{Eilenberg-MacLane-1953}, and are currently used in physics to describe spaces of indistinguishable particles \cite{Bloore-1980}. Details can be found in the very nice recent survey by S. Kallel \cite{Kallel-2024}. 

\item A Cauchy tensor product $\boxtimes$ in the category $\VB(\UConf(M))$ of vector bundles over $\UConf(M)$, which relies on the possibility to switch base points in $M$ and provides a symmetrized version of the usual external tensor product of vector bundles over $M$. The symmetric algebra with respect to $\boxtimes$ is a good candidate to describe multilocal observables.
An analog Cauchy tensor product exists for graded vector spaces \cite[Chapter 11.5]{Bourbaki-1989}, species, and other natural $\mathbb{S}$-modules \cite[Section 2.1]{Joyal-1986} \cite{Aguiar-Mahajan-2010}, but to our knowledge this version on vector bundles is new. 

\item The structure of a 2-monoidal category on $\VB(\UConf(M))$, where $\boxtimes$ is completed by a Hadamard tensor product $\otimes$ necessary to describe the Poisson algebra of observables induced by a kernel given on generators. For this, we introduce the new notion of a Poisson 2-algebra bundle. In the companion paper \cite{Frabetti-Kravchenko-Ryvkin-2024}, we show that the standard Poisson algebra structure of relativistic fields \cite[Section 4.4]{Rejzner-2016}, given by the Peierls bracket \cite{Peierls-1952}, is of this type. 
Symmetric $2$-monoidal categories have been introduced in \cite{Aguiar-Mahajan-2010} in the context of species. The idea that they are the appropriate categorical framework to describe Quantum Field Theory is due to S. Herscovich \cite{Herscovich-2019}, after R. Borcherds \cite{Borcherds-2011}. The classical Poisson algebra we introduce will lead to the same quantum version, via the quantization by Laplace pairing deformation as in \cite{Brouder-2009} \cite{Brouder-Fauser-Frabetti-Oeckl-2004}. 
A direct link to the species was made very recently in \cite{Norledge-2020}. 
\end{itemize}

The article is structured as follows. In Section \ref{sec:2} we study the category $\VB(\UConf(M))$ of vector bundles over the (unordered) configuration space $\UConf(M)$, starting from the basic definitions. In Subsection \ref{subsec:HC} we construct the new monoidal structure $\boxtimes$ on $\VB(\UConf(M))$ (called the Cauchy tensor product) and conclude that together with the usual tensor product $\otimes$ (which we refer to as the Hadamard tensor product), it forms a symmetric 2-monoidal category.\\

In Section \ref{sec:3} we describe algebra bundles with respect to these two monoidal structures. In Subsection \ref{subsec:Cauchy algebra bundles} we introduce our key examples, the Cauchy tensor algebra bundle $\bT^\boxtimes(V)$ and the Cauchy symmetric algebra bundle $\bS^{\boxtimes}(V)$ of a vector bundle $V\to M$, and show a striking effect of the Cauchy monoidal structure: the fibres of the vector bundle $\bS^\boxtimes(V)\to \UConf(M)$ are isomorphic to those of the exterior tensor product of $V$ over the collection of $k$-fold manifolds $M^k$, as vector spaces, even though the algebra bundle $\bS^\boxtimes(V)$ is commutative (of course with respect to the monoidal structure $\boxtimes$) while the collection of exterior tensor products of $V$ is not. Without changing the value of bundles, then the Cauchy monoidal structures allow us to consider symmetric sections of many points of $M$. 
In Subsection \ref{subsec:density} we show that the densities on $\UConf(M)$ are exactly the $\boxtimes$-symmetric algebra generated by the densities on $M$. 
In Subsection \ref{subsec:Cauchy-Hadamard 2-algebra bundles} we discuss 2-algebras, i.e. algebras with two multiplications with respect to the Cauchy and to the Hadamard tensor product, related by a natural distributional law induced by the interchange map of the $2$-monoidal category.\\

In Section \ref{sec:4} we discuss Poisson algebra bundles over $\UConf(M)$. 
We start by introducing the notion of a Poisson $2$-algebra bundle in Subsection \ref{subsec:Poisson 2-algebra}, which allows us to construct a canonical Poisson structure on the $2$-algebra bundle $\bS^\boxtimes S^\otimes (V)$ for any vector bundle $V$ on $M$ endowed with a skew-symmetric bundle map $k:V\boxtimes V\to \K$. 
In Subsection \ref{subsec:Poisson Cauchy algebra} we prove that this Poisson structure extends to the $\boxtimes$-algebra $\bS^\boxtimes S^\otimes(V) \otimes \Dens_{\UConf(M)}$, and in Subsection \ref{subsec:sections} we prove that it induces a Poisson algebra structure on its space of smooth sections. 
The paper ends with a concluding example coming from field theory which serves as a motivation for the second part \cite{Frabetti-Kravchenko-Ryvkin-2024} of our work: for $V=(JE)^*$, it is the bundle $\bS^\boxtimes S^\otimes (JE)^*\otimes \Dens_{\UConf(M)}$ on $\UConf(M)$ which carries the Poisson structure for observables of the fields $\varphi:M\to E$. The kernel $k$ of the Poisson bracket (or, more precisely, a distributional version of $k$) represents a pre-symplectic form on $\Gamma(M,E)$ usually given by the causal propagator defined by the Lagrangian density of the chosen field theory. In \cite{Frabetti-Kravchenko-Ryvkin-2024} we discuss which is the good space of distributional sections of this Poisson bundle which should be taken to represent off-shell observables of relativistic fields. 


\section{The $2$-monoidal category of vector bundles over configuration spaces}
\label{sec:2}

The external tensor product of vector bundles over a manifold $M$ gives vector bundles over the cartesian product $M\times M$, and therefore it does not preserve the base manifold. 
To have a single comprehensive base manifold to work on, one can consider the disjoint union $\bigsqcup_{k\in\mathbb N_0} M^k$, which is a manifold of non-pure dimension. 
If, moreover, the sections one is interested in are to be symmetric in the different $M$-entries, one should consider the disjoint union $\bigsqcup_{k\in\mathbb N_0} M^k/S_k$ of the orbits under the action of the symmetric groups. 
This space, however, is an orbifold (a manifold with singularities) and its differential calculus requires a specific attention. To avoid complications due to singularities, in this paper we consider the disjoint union of configuration spaces, which are manifolds, as a comprehensive environment to work on. 
In this Section we show that the natural extension of the usual (internal) and of the external tensor products to vector bundles over the configuration space define a category with extremely good properties. 

\subsection{Configuration space of a manifold}
\label{subsec:Conf(M)}

Let $M$ be a smooth manifold. 
A \emph{configuration of $k$ points of $M$}, or \emph{$k$-point configuration}, is a finite set $\underline{x}=\{x_1,...,x_k\}$ of $k$ distinct points of $M$. Denote by $|\underline{x}|=k$ the number of points contained in a configuration $\underline{x}$. 
For $k=0$, the only zero-point configuration is the empty set $\emptyset$ (which is the only subset of $M$ of cardinality $0$). We call this configuration the \emph{vacuum}. 
Denote by $\UConf_k(M)$ the set of $k$-point configurations of $M$ and by 
\begin{align}\label{eq:Conf(M)}
 \UConf(M) = \bigsqcup_{k=0}^\infty \UConf_k(M)
\end{align}
the set of all configurations in $M$, i.e. the set of finite subsets of $M$. 
The space $\UConf(M)$ is called the (unordered) \emph{configuration space} of $M$ \cite[Definition 4.3]{May-1972}.\footnote{The same set furnished with a different topology is sometimes referred to as the \emph{Ran space} of $M$ cf. e.g. \cite{Lurie-2017}.} We quickly recall how it is topologized:

\begin{lemma}[cf. e.g. Definition 4.3 in \cite{May-1972}] \label{lem:Conf(M) manifold}
The configuration space $\UConf(M)$ is a smooth manifold of non-pure dimension, with the $k$-th component locally diffeomorphic to the collection of $k$-fold products $M^k$. 
Moreover, the identity map on $M=\UConf_1(M)$ gives an embedding
\begin{align}\label{eq:inclusion}
 i:M \hookrightarrow \UConf(M). 
\end{align} 
\end{lemma}

\begin{proof}
The set of \emph{$k$-point ordered configurations} of $M$ is the open and dense submanifold of $M^k$
\begin{align*}
 \OConf_k(M) = M^k\setminus \Delta^{(k)}
\end{align*}
where 
\begin{align*}
 \Delta^{(k)} = \{(x_1,...,x_k)\in M^k\ |\ x_i=x_j \ \mbox{for some}\ i\neq j\}
\end{align*}
is the \emph{fat diagonal} of $M^k$. 
The space $\OConf_k(M)$ contains (ordered) $k$-tuples $\vec{x}=(x_1,...,x_k)\in M^k$ of distinct points, in contrast to $\UConf_k(M)$ which contains (unordered) sets $\underline{x}=\{x_1,\dots,x_k\}\subset M$ of $k$ distinct points. 
The symmetric group $S_k$ acts on $\OConf_k(M)$ by permuting the points in the $k$-tuples and 
\begin{align*}
 \UConf_k(M)=\OConf_k(M)/S_k
\end{align*} 
is the quotient space under this action. 
Since the action of $S_k$ on $\OConf_k(M)$ is free, the quotient space $\UConf_k(M)$ is a smooth manifold, locally diffeomorphic to $\OConf_k(M)$ and therefore to $M^k$. 
We topologize $\UConf(M)$ as the countable disjoint union of the $\UConf_k(M)$, i.e. as in Equation \eqref{eq:Conf(M)}. Hence, $\UConf(M)$ is a countable disjoint union of smooth manifolds, hence a manifold itself (of non-pure dimension). 
\end{proof}

For any $k\geq 0$, we denote by $\sigma\cdot \vec{x}=(x_{\sigma^{-1}(1)},...,x_{\sigma^{-1}(k)})$ the left action of $\sigma\in S_k$ on $\vec{x}=(x_1,...,x_k)\in \OConf_k(M)$, and by 
\begin{align*}\label{eq:q}
 q_k:\OConf_k(M)\to \UConf_k(M)=\OConf_k(M)/S_k 
\end{align*}
the canonical quotient map, and by $q:\OConf(M)\to \UConf(M)$ the induced surjective map on the full configuration spaces. One way to construct a chart $(U, \Phi)$ near $\underline x=\{x_1,...,x_k\}\in \UConf(M)$ is by taking disjoint open charts $(U_1,\phi_1),...,(U_k,\phi_k)$ around $x_1,...,x_k$ and set 
\begin{equation} \label{eq:chartconf}
 \Phi = (\phi_1,...,\phi_k)\circ (q_k|_{U})^{-1} 
\end{equation}
with $U=q_k(U_1\times \cdots \times U_k)$. The invertibility of $q_k$ on $U$ follows from the disjointness of $U_1,..., U_k$.

\begin{remark}\label{rem:noncompact}
For any $k\geq 0$, the quotient map $q_k$ is a local diffeomorphism, therefore $\UConf_k(M)$ has the same \emph{local} properties of $\OConf_k(M)$ and therefore also of $M^k$, for instance its dimension is $k \cdot \dim(M)$. 
However, $\UConf_k(M)$ does not have the same \emph{global} properties of $\OConf_k(M)$ or of $M^k$. The two main differences are:
\begin{itemize}
\item 
If a manifold $M$ is orientable, then $M^k$ and $\OConf_k(M)$ are also orientable, however its $k$-point configuration space $\UConf_k(M)$ is not necessarily orientable. For instance, the circle $S^1, \ S^1\times S^1, \ \OConf_2(S^1)\cong S^1\times \R$ - are all orientable, while $\UConf_2(S^1)$ forms the M\"obius strip and hence it is not orientable.

\item If $M$ is compact, then $M^k$ is compact, but $\OConf_k(M)$ and $\UConf_k(M)$ are not, since taking out the diagonals makes them open. 
\end{itemize} 
\end{remark}

\begin{remark} \label{rem:OConf_k(M) cover space}
If $M$ is a connected manifold, the quotient map $q_k$ is a covering space of degree $k!$, since the symmetric group $S_k$ acts freely on each fibre.
\end{remark}

The classical constructions available in the category of manifolds can still be carried out on $\UConf(M)$, by working separately on each $k$-point component.
In particular, in this paper we consider vector bundles on $\UConf(M)$. 
The embedding $i:M\hookrightarrow \UConf(M)$ induces standard restriction and extension maps on vector bundles, on their sections (also on those with compact support) and on distributions, which are key ingredients of the covariant description of the Poisson structure of observables in field theory, which we treat in a subsequent paper. 

Note that there is no natural map $\UConf(M)\to M$: already the zero-point configuration $\UConf_0(M) = \{\emptyset\}$ has no natural counterpart in $M$. Indeed, one of the key properties of a configuration space is to be a \emph{pointed manifold}, with base point given by the vacuum $\emptyset$. 


\subsection{Vector bundles over configuration spaces}
\label{subsec:VB}

Let $\K$ be the field $\R$ or $\C$. 
A \emph{$\K$-vector bundle over $\UConf(M)$}, denoted $\pi:\bV\to \UConf(M)$, is a collection of $\K$-vector spaces $\bV_{\underline{x}}$ above each configuration of points of $M$, called the \emph{fibres}, verifying the usual requirements of vector bundles (local triviality inducing linear isomorphisms of the fibres). 
Since local triviality makes sense only on connected components of $\UConf(M)$, we can set $\bV_k = \bigsqcup_{\underline{x}\in \UConf_k(M)} \bV_{\underline{x}}$ and regard the vector bundle $\bV$ as a collection of $\K$-vector bundles $\pi_k:\bV_k\to \UConf_k(M)$ with possibly different ranks, called \emph{(homogeneous) components} of $\bV$. Finally we have
\begin{align*} 
 \bV = \bigsqcup_{\underline{x}\in \UConf(M)} \bV_{\underline{x}} = \bigsqcup_{k=0}^\infty \bV_k . 
\end{align*}
Let us call $\bV$ \emph{connected} if $\bV_0=\bV_\emptyset=\K$.
Naturally, a \emph{bundle map $\bV\to \bW$} over $\UConf(M)$ is a map commuting with the projections, therefore it is given by a collection of usual bundle maps $\bV_k\to \bW_k$ over $\UConf_k(M)$ for each $k$. 
Therefore, as usual, a bundle map can be multiplied by scalars, and the $\K$-vector space of bundle maps from $\bV$ to $\bW$ is the product of vector spaces
\begin{align} \label{eq:bundle map}
 \Mor(\bV,\bW) \cong \prod_k \Mor(\bV_k,\bW_k). 
\end{align}
All usual constructions on vector bundles over a manifold can be pushed to vector bundles over its configuration space, which then form a category with the operations and the properties of usual vector bundles. 
Let us denote by $\VB(M)$ the category of $\K$-vector bundles over $M$, and by $\VB(\UConf(M))$ the category of $\K$-vector bundles over $\UConf(M)$. 

In particular, in $\VB(\UConf(M))$ there is an internal \emph{hom-bundle} of fibrewise linear maps 
\begin{align*}
 \Hom(\bV,\bW) = \bigsqcup_{\underline{x}\in \UConf(M)} \Hom(\bV_{\underline{x}},\bW_{\underline{x}}) = \bigsqcup_{k=0}^\infty \Hom(\bV_k,\bW_k), 
\end{align*}
whose smooth sections are in bijection with $\Mor(\bV,\bW)$. 
Given a vector space $V$, the \emph{trivial vector bundle with fibre $V$} is the cartesian product $\UConf(M)\times V$. 
Then, the \emph{trivial line bundle} is $\UConf(M)\times \K$, and the \emph{dual vector bundle} of $\bV$ is 
\begin{align*}
 \bV^* &:= \Hom(\bV,\UConf(M)\times \K) \cong 
 \bigsqcup_{\underline{x}} \bV_{\underline{x}}^* = \bigsqcup_{k} \bV_k^*. 
\end{align*}
Finally, $\bV$ is a \emph{subbundle} of $\bW$ if $\bV_k$ is a subbundle of $\bW_k$ for all $k$. \\

\begin{remark}
All the vector bundles on $\UConf(M)$ we are interested in ultimately come from a given vector bundle $E$ on $M$. Even if $E$ has finite rank, many of the associated objects we will be interested in and treat in the sequel, such as the jet bundle $JE$ and the tensor bundle $T^\otimes(E)$, will have infinite rank (i.e. infinite dimension in each fibre). However, in all cases we will be able to work with them as if they were of finite rank. Following \cite{Blohmann-2023}, we distinguish two cases of infinite-rank vector bundles (dual to each other in the categorical sense):

\begin{itemize}
\item \emph{ind-finite vector bundles}: These are filtered (inductive) colimits of finite-rank vector bundles. For instance, the \emph{tensor bundle} $T^\otimes(E)=\bigoplus_{n=0}^\infty E^{\otimes n}$ is the inductive colimit of the sequence of injective bundle maps 
$$
T^{\otimes \leq 0}(E)\hookrightarrow T^{\otimes \leq 1}(E)\hookrightarrow T^{\otimes \leq 2}(E)\hookrightarrow \cdots
$$
where $T^{\otimes \leq k}(E)=\bigoplus_{n=0}^k E^{\otimes n}$. 
Indeed, any element of this colimit lies in fact in a finite-rank term $T^{\otimes \leq k}(E)$. When we want to add vectors or carry out other operations among tensors of different length, we can just find a tensor power $k$ big enough, so that $T^{\otimes \leq k}(E)$ contains all of them. 

\item \emph{pro-finite vector bundles}: These are cofiltered (projective) limits of finite-rank vector bundles. For instance, the \emph{infinite jet bundle} $JE$ is by definition the projective limit of the sequence of projective bundle maps
$$
J^0E\twoheadleftarrow J^1E\twoheadleftarrow J^2E \twoheadleftarrow \cdots
$$
where $J^kE$ is the bundle over $M$ with fibre at $x\in M$ given by the $k$-jets $j^k_x\varphi$ at $x$ of all smooth local sections $\varphi:U_x\to E$ defined on an open neighborhood of $x$ (that is, the equivalence class of such local sections under the equivalence relation of \emph{contact of order $k$ at $x$}). 
An element in the infinite jet bundle is at once the set of all derivatives of a smooth local section at a given point of $M$, and naturally projects to all $k$-order derivatives, that is, to all $k$-jet bundles, in bundle-compatible ways. 
All operations among jets of different orders in fact take place in the infinite jet bundle $JE$ and are then projected to the suitable finite-order jet bundle $J^kE$ we are interested in. 

Another example of a pro-finite vector bundle is the \emph{completed tensor bundles} $\hat{T}^{\otimes}(E)$, which can be defined as the infinite product $\prod_{n=0}^\infty E^{\otimes n}$, or by bundle duality as $(T(E^*))^*$.
\end{itemize}
The dual bundle of a pro-bundle is an ind-bundle and vice versa.
Iterating several pro- (or ind-) constructions is not a problem, subtleties would only occur if we try mixing pro- and ind- constructions, which we will avoid. 
\end{remark}

The inclusion $i:M\hookrightarrow \UConf(M)$ induces two functors on vector bundles: The usual pullback $i^*:\VB(\UConf(M)) \to \VB(M)$, which brings a vector bundle $\bV\to \UConf(M)$ to its restriction $i^*\bV=\bV_1\to M$ and the pushforward \footnote{This is not properly a pushforward of vector bundles over manifolds, but it is a \emph{diffeological pushforward} in the sense of \cite[Section 3]{Wu-2023}.}
$i_*:\VB(M)\to \VB(\UConf(M)$ given, for any vector bundle $V \to M$, by
\begin{align}\label{eq:pullback pushforward}
i_*V =
\begin{cases}
 V & \text{over $M$} \\ 
 \UConf_k(M) \times \{0\} & \text{over $\UConf_k(M)$ for $k\neq 1$} 
\end{cases} 
\end{align}
The pullback is a left inverse of the pushforward, that is, $i^* i_* V=V$ for any vector bundle $V$ on $M$. Moreover, the pushforward $i_*$ is fully faithful, therefore $\VB(M)$ is a full subcategory of $\VB(\UConf(M))$, and $i_*$ is injective on vector bundles (up to isomorphisms).

\begin{example}
\label{ex:trivial bundle}
The trivial line bundle on $\UConf(M)$ is the usual cartesian product $\UConf(M)\times \K$. 
Its pullback on $M$ by $i$ coincides with the trivial line bundle on $M$, that is, 
\begin{align*}
 i^*(\UConf(M)\times \K) = M\times \K. 
\end{align*}
However, the pushforward of the trivial line bundle on $M$ is not the trivial line bundle on $\UConf(M)$, it is
\begin{align*}
 i_*(M\times \K) = \begin{cases}
 M\times \K & \text{over $M$} \\ 
 \UConf_k(M) \times \{0\} & \text{over $\UConf_k(M)$ for $k\neq 1$} . 
 \end{cases}
\end{align*}
\end{example}

\begin{example}
\label{ex: T(Conf(M))}
Let $TM\to M$ be the tangent bundle of $M$. 
Since $\UConf_k(M)$ is locally diffeomorphic to $\OConf_k(M)$ (and hence to $M^k$), the \emph{tangent bundle} of $\UConf(M)$ is the bundle 
\begin{align*}
 \bT\UConf(M) &= 
 \begin{cases}
 \{\emptyset\} \times \{0\} & \text{over $\UConf_0(M)=\{\emptyset\}$} \\ 
 TM & \text{over $\UConf_1(M)=M$} \\ 
 \displaystyle \bigsqcup_{\{x_1,...,x_k\}} (T_{x_1}M \oplus\cdots\oplus T_{x_k}M) & \text{over $\UConf_k(M)$ for $k\geq 2$} . 
 \end{cases}
\end{align*}
Therefore, the bundle $i_*TM$ over $\UConf(M)$, which is zero on $\UConf_{k\geq 2}(M)$, is a proper subbundle of $\bT(\UConf(M))$. 
\end{example}

\begin{remark} \label{rem:VB on OConf(M)}
The description of vector bundles over $\UConf(M)$ given above can be repeated word by word to describe vector bundles over $\OConf(M)$, simply replacing $k$-configurations $\underline{x}$ by $k$-tuples $\vec{x}$ of distinct points of $M$, and gives rise to the category $\VB(\OConf(M))$ of vector bundles over $\OConf(M)$. 
Vector bundles over $\UConf(M)$ and over $\OConf(M)$ are then related by the pullback functor 
\begin{align*}
 q^*:\VB(\UConf(M))\to \VB(\OConf(M))
\end{align*}
along the quotient map $q:\OConf(M)\to \UConf(M)$: the pullback $q^*(\bV)$ of a vector bundle $\bV$ on $\UConf(M)$ has a fibre $q^*(\bV)_{\vec{x}} = \bV_{q(\vec{x})}$ above any point $\vec{x}\in \OConf(M)$. This pullback restricts to the homogeneous components as a functor $q_k^*:\VB(\UConf_k(M))\to \VB(\OConf_k(M))$ for any $k\geq 0$. For brevity, we will write $q_k^*(\bV)$ for $q_k^*(\bV_k)$. 
\end{remark}


\subsection{Hadamard and Cauchy tensor products of vector bundles}
\label{subsec:HC}

We denote by $\oplus$ the direct sum in $\VB(\UConf(M))$ given by the usual direct sum of vector bundles over each component $\UConf_k(M)$. 

\begin{definition}\label{def:hadamard cauchy}
For $\bV,\bW\in \VB(\UConf(M))$, we call \emph{Hadamard tensor product} the usual tensor product $\bV\otimes \bW$ of vector bundles, whose fibre over a configuration $\underline{x}$ is given by
\begin{align*}
 (\bV\otimes \bW)_{\underline x} = \bV_{\underline x}\otimes \bW_{\underline x}. 
\end{align*}
We call \emph{Cauchy tensor product} the collection $\bV\boxtimes \bW = \bigsqcup_{\underline{x}\in \UConf(M)} (\bV\boxtimes \bW)_{\underline x}$ of fibres
\begin{align*}
 (\bV\boxtimes \bW)_{\underline x}=\bigoplus_{\underline x=\underline x'\sqcup \underline x''} (\bV_{\underline x'}\otimes \bW_{\underline x''}), 
\end{align*}
where the sum is over all (ordered) splits of the configuration $\underline x$ in two disjoint configurations, including the zero-point configuration. 
\end{definition}

Note that the \emph{ordered} splits $\underline{x}'\sqcup \underline{x}''$ and $\underline{x}''\sqcup \underline{x}'$ are different and must both be considered in the sum.
For instance:
\begin{align*}
 (\bV\boxtimes \bW)_{\emptyset} 
 &= \bV_{\emptyset} \otimes \bW_{\emptyset} \\ 
 (\bV\boxtimes \bW)_{x} 
 &= (\bV_{x} \otimes \bW_{\emptyset}) \oplus (\bV_{\emptyset} \otimes \bW_{x}) \\ 
 (\bV\boxtimes \bW)_{\{x_1,x_2\}} 
 &= (\bV_{\{x_1,x_2\}}\otimes \bW_{\emptyset}) \oplus (\bV_{x_1} \otimes \bW_{x_2}) \oplus (\bV_{x_2} \otimes \bW_{x_1}) \oplus (\bV_{\emptyset} \otimes \bW_{\{x_1,x_2\}}) \\ 
 (\bV\boxtimes \bW)_{\{x_1,x_2,x_3\}} 
 &= (\bV_{\{x_1,x_2,x_3\}}\otimes \bW_{\emptyset}) \\ 
 & \hspace{0.5cm} \oplus (\bV_{\{x_1,x_2\}} \otimes \bW_{x_3}) \oplus (\bV_{\{x_1,x_3\}} \otimes \bW_{x_2}) \oplus (\bV_{\{x_2,x_3\}} \otimes \bW_{x_1}) \\ 
 & \hspace{0.5cm} \oplus (\bV_{x_1} \otimes \bW_{\{x_2,x_3\}}) \oplus (\bV_{x_2} \otimes \bW_{\{x_1,x_3\}}) \oplus (\bV_{x_3} \otimes \bW_{\{x_1,x_2\}}) \\ 
 & \hspace{0.5 cm} \oplus (\bV_{\emptyset} \otimes \bW_{\{x_1,x_2,x_3\}}). 
\end{align*}

The Hadamard tensor product $\bV \otimes \bW$ preserves the $k$-components and therefore its smooth vector bundle structure is a standard construction in geometry. For $\bV \boxtimes \bW$ the smooth vector bundle structure is not immediate and we will construct it in Theorem \ref{thm:Cauchy vs external} below. To do so, we will lift it to the ordered configuration space $\OConf(M)$, because (unlike $\UConf(M)$) it has natural deconcatenation maps $\OConf_{k+l}(M)\to \OConf_{k}(M)\times \OConf_{l}(M)$.

\begin{lemma} \label{lem:external tensor product on OConf(M)}
The external tensor product of vector bundles is  well-defined as an operation 
\begin{align*}
 \exttens: \VB(\OConf(M)) \times \VB(\OConf(M) \to \VB(\OConf(M)). 
\end{align*}
\end{lemma}

\begin{proof}
For any integers $i,j\geq 0$, the external tensor product (which we denote by $\exttens$ to distinguish it from the Cauchy tensor product) can be seen as an operation 
\begin{align*}
 \exttens:\VB(\OConf_i(M)) \times \VB(\OConf_j(M)) \to \VB(\OConf_{i}(M)\times \OConf_{j}(M)) 
\end{align*}
which turns two vector bundles $\bV\to \OConf_{i}(M)$ and $\bW\to \OConf_{j}(M)$ to the vector bundle $\bV\exttens \bW\to \OConf_{i}(M)\times \OConf_{j}(M)$ with fibre 
\begin{align*}
 (\bV \exttens \bW)_{(x_1,...,x_{i+j})}= \bV_{(x_1,...,x_i)} \otimes \bW_{(x_{i+1},...,x_{i+j})}. 
\end{align*}
The space $\OConf_{{i+j}}(M)\subset M^{i+j}$ is a subspace of $\OConf_{i}(M)\times \OConf_{j}(M)\subset M^i\times M^j=M^{i+j}$. By abuse of notation the restriction of $\bV \exttens \bW$ to $\OConf_{{i+j}}(M)$ is denoted also by $\bV \exttens \bW$.
\end{proof}

For any permutation $\sigma\in S_k$ and any vector bundle $\bV_k\to \OConf_k(M)$, let $\sigma^* \bV_k$ be the pull back of $\bV_k$ by the action of $\sigma$, whose fibres are
\begin{align*}
 (\sigma^* \bV_k)_{\vec{x}} = (\bV_k)_{\sigma \cdot \vec{x}} = \bV_{(x_{\sigma^{-1}(1)},...,x_{\sigma^{-1}(k)})} 
\end{align*}
For any $i+j=k$, denote by $\USh(i,j)$ the set of $(i,j)$-unshuffle permutations, that is, the $k$-permutations $\sigma$ such that $\sigma^{-1}(1)<...<\sigma^{-1}(i)$ and $\sigma^{-1}(i+1)<...<\sigma^{-1}(i+j)$, cf. e.g. \cite{Eilenberg-MacLane-1953}.

\begin{theorem} \label{thm:Cauchy vs external}
For any vector bundles $\bV$ and $\bW$ on $\UConf(M)$, there is a unique vector bundle structure on the Cauchy tensor product $\bV\boxtimes \bW$ which makes it a symmetric version of the external tensor product on $\OConf(M)$, in the sense that for any $k\geq 0$ there is an isomorphism 
\begin{align*}
 q_k^*(\bV\boxtimes \bW) & \cong \bigoplus_{I\sqcup J=\{1,...,k\}} q_{|I|}^*(\bV)\exttens q_{|J|}^*(\bW) \\ 
 & \cong \underset{\sigma\in \USh(i,j)}{\bigoplus_{i+j=k}} \sigma^*\big(q_i^*(\bV) \exttens q_j^*(\bW)\big) 
\end{align*}
of vector bundles over $\OConf_k(M)$. 
\end{theorem} 

\begin{proof}
For any ordered $k$-point configuration $\vec{x}=(x_1,...,x_k)$ with $q(\vec{x})=\{x_1,...,x_k\}=\underline{x}$, we have 
\begin{align*}
 q_k^*(\bV\boxtimes \bW)_{\vec{x}} &= (\bV\boxtimes \bW)_{\underline{x}} = \bigoplus_{\underline{x}=\underline{x}'\sqcup \underline{x}''} \bV_{\underline{x}'} \otimes \bW_{\underline{x}''}. 
\end{align*}
On the right hand-side, the split of the configuration $\underline{x}$ in two disjoint subconfigurations $\underline{x}'$ and $\underline{x}''$ corresponds to the split of the $k$-tuple $\vec{x}$ in two disjoint tuples $\vec{x}'$ and $\vec{x}''$ of all possible lengths $i$ and $j$ with $i+j=k$, plus all possible choices of $i$ points among $k$, which are counted by the binomial coefficient $\binom{k}{i}$. Since this binomial coefficient is also the cardinality of $(i,j)$-shuffle permutations, there is a bijection of vector spaces 
\begin{align*}
 q_k^*(\bV\boxtimes \bW)_{\vec{x}} 
 &\cong \underset{\sigma\in \USh(i,j)}{\bigoplus_{i+j=k}} \big(q_i^*(\bV_i)\exttens q_j^*(\bW_j)\big)_{\sigma\cdot \vec{x}} \\
 &= \underset{\sigma\in \USh(i,j)}{\bigoplus_{i+j=k}} \Big(\sigma^* \big(q_i^*(\bV_i)\exttens q_j^*(\bW_j)\big)\Big)_{\vec{x}}
\end{align*}
over each $\vec{x}\in \OConf_k(M)$. 
This gives a bijection of vector bundles over $\OConf_k(M)$ 
\begin{align*}
 q_k^*(\bV\boxtimes \bW) &\cong \underset{\sigma\in \USh(i,j)}{\bigoplus_{i+j=k}} \sigma^* \big(q_i^*(\bV)\exttens q_j^*(\bW)\big). 
\end{align*}
The same bijection can be expressed in terms of the choice of the multiindex $I\subset \{1,...,k\}$: if we denote by 
\begin{align*}
 \pi_I: \OConf_{k}(M)\to \OConf_{|I|}(M)
\end{align*}
the projection to the components corresponding to the multiindex $I$ with cardinality $|I|$, over $\OConf_k(M)$ the above bijection of vector bundles becomes 
\begin{align}
q_k^*(\bV\boxtimes \bW) & \cong
\bigoplus_{I\sqcup J=\{1,...,k\}} \pi_I^*q_{|I|}^*(\bV_{|I|}) \otimes \pi_J^*q_{|J|}^*(\bW_{|J|}) \label{eq:liftiso}\\ 
& \cong \bigoplus_{I\sqcup J=\{1,...,k\}} q_{|I|}^*(\bV_{|I|}) \exttens q_{|J|}^*(\bW_{|J|}). \nonumber
\end{align}
Since $q:\OConf_k(M)\to \UConf_k(M)$ is a local diffeomorphism, there is a unique smooth structure on $\bV\boxtimes \bW$ which makes this bijection a smooth vector bundle isomorphism. 
\end{proof}

To obtain a local trivialization of $\bV\boxtimes \bW$ near $\underline x=\{x_1,...,x_k\}\in \UConf(M)$, we pick a chart $q(U)=q(U_1\times \cdots \times U_k)$ on $\UConf_k(M)$ as in \eqref{eq:chartconf}, with the additional condition that $\bV,\bW$ are trivializable over $q(U_{i_1}\times \cdots \times U_{i_l})$ for all $\{i_1,..., i_l\}\subset \{1,...,k\}$. Then over $U=U_1\times\cdots\times U_k$, the isomorphism \eqref{eq:liftiso} provides a trivialization of $\bV\boxtimes \bW|_{q(U)}\cong q_k^*(\bV\boxtimes \bW)|_U$.

\begin{lemma}\label{lem:monoidal}
The Hadamard and the Cauchy tensor products give two symmetric monoidal structures on the category $\VB(\UConf(M))$ (in the sense of \cite{MacLane-1998}), such that
\begin{itemize}
\item the unit for $\otimes$ is given by the trivial line bundle $\bI_\otimes = \UConf(M)\times \K$ and the braiding $\tau_\otimes:\bV\otimes \bW \overset{\cong}{\to} \bW\otimes \bV$ is given by $\tau_\otimes(v_{\underline{x}}\otimes w_{\underline{x}}) = w_{\underline{x}}\otimes v_{\underline{x}}$ for any $\underline{x}\in \UConf(M)$, 

\item the unit for $\boxtimes$ is given by the bundle 
\begin{align*}
 \bI_\boxtimes &= 
 \begin{cases}
 \{\emptyset\}\times \K & \text{over $\UConf_0(M)$}\\
 \UConf_k(M)\times \{0\} & \text{over $\UConf_k(M)$ for $k\geq 1$}
\end{cases} 
\end{align*}
and the braiding $\tau_\boxtimes:\bV\boxtimes \bW \overset{\cong}{\to} \bW\boxtimes \bV$ is given by $\tau_\boxtimes(v_{\underline{x}'}\boxtimes w_{\underline{x}''}) = w_{\underline{x}''}\boxtimes v_{\underline{x}'}$ for any disjoint configurations $\underline{x}',\underline{x}''\in \UConf(M)$. 
\end{itemize}
\end{lemma}

\begin{proof}
The Hadamard tensor product $\otimes$ restricts to each $k$-point configuration manifold and coincides with the usual tensor product of vector bundles in the category $\VB(\UConf_k(M))$, so it is a symmetric monoidal structure. The Cauchy tensor product $\boxtimes$ distributes the fibres of two vector bundles over the points of each configuration, in all possible ways: the axioms of a symmetric monoidal structure are easily verified. 
\end{proof}

\begin{remark} \label{rem:graded braiding}
Vector bundles on $\UConf(M)$ can be graded. Then, the braidings $\tau_\otimes$ and $\tau_\boxtimes$ naturally extend to graded versions $\tau_\otimes^\sgn$ and $\tau_\boxtimes^\sgn$ with the usual Koszul sign convention: if $\deg(v_{\underline{x}})$ denotes the degree of the vector $v_{\underline{x}}$ belonging to a graded vector bundle, then we set 
\begin{align*}
 & \tau_\otimes: \bV\otimes \bW \to \bW \otimes \bV,\ v_{\underline{x}}\otimes w_{\underline{x}} \mapsto (-1)^{\deg(v_{\underline{x}}) \deg(w_{\underline{x}})} w_{\underline{x}}\otimes v_{\underline{x}} \\ 
 & \tau_\boxtimes: \bV\boxtimes \bW \to \bW \boxtimes \bV,\ v_{\underline{x}'}\otimes w_{\underline{x}''} \mapsto (-1)^{\deg(v_{\underline{x}'}) \deg(w_{\underline{x}''})} w_{\underline{x}''}\otimes v_{\underline{x}'}. 
\end{align*}
These braidings provide a usual trick to define alternating maps on non-graded vector bundles, by putting them in degree $1$. 
\end{remark}

Notice that the unit $\bI_\otimes$ with respect to the Hadamard monoidal structure restricts to the unit $I=M\times \K$ of the usual tensor product of vector bundles on $M$, via the pullback $i^*$ of the embedding $i:M\hookrightarrow \UConf(M)$. 
However the unit $\bI_\boxtimes$ with respect to the Cauchy monoidal structure corresponds to the zero bundle on $M$, via the pullback $i^*$, and is not related to the pushforward of the unit $I$ on $M$, described in Example \ref{ex:trivial bundle}. 

The monoidal structures $\otimes$ and $\boxtimes$ are compatible in the sense that they form a \emph{symmetric $2$-monoidal category} (cf. e.g. \cite{Aguiar-Mahajan-2010, Herscovich-2019}), that is, there exists the following data:
\begin{itemize}
\item 
three morphisms $\mu: I_\otimes\boxtimes I_\otimes\to I_{\otimes}$, $\Delta: I_\boxtimes\to I_\boxtimes\otimes I_\boxtimes$ and $\nu: I_\boxtimes\to I_\otimes$ such that $(I_\otimes,\mu,\nu)$ is a monoid in the monoidal category $(\boxtimes,I_\boxtimes)$ and $(I_\boxtimes,\Delta,\nu)$ is a comonoid in the monoidal category $(\otimes,I_\otimes)$, 
\item 
an interchange map $\zeta: (A\otimes B)\boxtimes (C\otimes D)\to\ (A\boxtimes C)\otimes (B\boxtimes D)$ which distributes appropriately over triple products, 
\end{itemize}
satisfying some compatibility conditions which can be found in \cite[Section 6.1]{Aguiar-Mahajan-2010} or \cite[Definitions 3.2.2]{Herscovich-2019}. 

\begin{theorem}\label{thm: VB(Conf(M)) is 2-monoidal}
The category $\VB(\UConf(M))$ is a symmetric 2-monoidal category when equipped with:
\begin{itemize}
\item map $\mu: \bI_\otimes\boxtimes \bI_\otimes\to \bI_{\otimes}$ given by the multiplications $\K_{\underline {x}'}\otimes \K_{\underline {x}''}\to \K_{\underline {x}'\sqcup \underline {x}''}$, 

\item map $\Delta : \bI_\boxtimes\to \bI_\boxtimes\otimes \bI_\boxtimes$ given by the usual isomorphism $\K_{\emptyset}\otimes \K_{\emptyset}\to \K_{\emptyset}$ over the zero-point configuration and trivial elsewhere, 

\item map $\nu: \bI_\boxtimes\to \bI_\otimes$ given by the inclusion of $\bI_\boxtimes=\K_{\emptyset}$ into $\bI_\otimes$, 

\item interchange map 
$$
\zeta: (\bV'\otimes \bW')\boxtimes (\bV''\otimes \bW'')\to\ (\bV'\boxtimes \bV'')\otimes (\bW'\boxtimes \bW'')
$$
given by the inclusion 
$$
\bigoplus_{\underline x'\sqcup \underline x''=\underline x} \bV'_{\underline x'}\otimes \bW'_{\underline x'}\otimes \bV''_{\underline x''}\otimes \bW''_{\underline x''} 
\hookrightarrow
\bigoplus_{\underline x'\sqcup \underline x''=\underline {\tilde x}'\sqcup \underline {\tilde x}''} \bV'_{\underline {x}'}\otimes \bV''_{\underline {x}''}\otimes \bW'_{\underline {\tilde x}'}\otimes \bW''_{\underline {\tilde x}''}.
$$
\end{itemize}
\end{theorem}

\begin{proof}
The proof goes by direct verification of the axioms. We note that while all involved operations and morphisms have to be (and are) smooth, the verification of the identities can be done fibre by fibre, since the equations are satisfied, once they are satisfied on each fibre separately. The verification of the axioms fibrewise can be carried out by direct computation, it also follows from \cite[Proposition 8.68]{Aguiar-Mahajan-2010}, where it is formulated in terms of species.
\end{proof}

\begin{remark}
In fact, this category is even \emph{framed monoidal} in the sense of \cite[Definition 3.2.6]{Herscovich-2019}, with framing given by the forgetful functor
$(\VB(\UConf(M)),\boxtimes,\otimes) \to (\VB(\UConf(M)),\boxtimes)$, and coherence morphisms $\psi_0=\mathrm{id}$, $\psi_2=\mathrm{id}$, $\phi_0=\nu$ and $\phi_2:\bV\boxtimes \bW\to \bV\otimes \bW$ given by the isomorphism over the zero-point configuration $\emptyset\in\UConf(M)$ and trivial elsewhere.
\end{remark}

\begin{lemma} \label{lem:dual VB monoidal}
On finite-dimensional vector bundles, taking the dual vector bundle is a monoidal functor on \cite{Aguiar-Mahajan-2010, Herscovich-2019}: $\VB(\UConf(M))$ for both the Hadamard and the Cauchy tensor products, that is, 
\begin{align*}
 (\bV \otimes \bW)^* &\cong \bV^* \otimes \bW^* \\ 
 (\bV \boxtimes \bW)^* &\cong \bV^* \boxtimes \bW^*
\end{align*}
for any (component-wise finite-dimensional) vector bundles $\bV$ and $\bW$.
\end{lemma}

When $\bV$, $\bW$ are infinite-dimensional the tensor products have to be appropriately completed.

\begin{proof}
This follows from the fact that the dual vector bundle is the bundle of dual components, which are usual vector bundles of finite rank. For the Hadamard tensor product the isomorphism is immediate. For the Cauchy tensor product, we have 
\begin{align*}
 (\bV\boxtimes \bW)^* 
 &= \bigsqcup_{\underline{x}\in \UConf(M)} \left(\bigoplus_{\underline{x}=\underline{x}'\sqcup \underline{x}''} (\bV_{\underline{x}'}\otimes \bW_{\underline{x}''})\right)^* 
 \cong \bigsqcup_{\underline{x}\in \UConf(M)} \bigoplus_{\underline{x}=\underline{x}'\sqcup \underline{x}''} (\bV_{\underline{x}'}^*\otimes \bW_{\underline{x}''}^*) 
 = \bV^*\boxtimes \bW^* . 
\end{align*}
\end{proof}

\begin{remark}\label{rk:i* monoidal Hadamard}
Using the pullback and the pushforward along the embedding $i:M\hookrightarrow \UConf(M)$, we can compare the tensor product of vector bundles over $M$ and over $\UConf(M)$: 
\begin{itemize}
\item Both functors, the pullback $i^*$ and the pushforward $i_*$, are monoidal with respect to the tensor product in $\VB(M)$ and to the Hadamard tensor product in $\VB(\UConf(M))$, that is, 
$i^*(\bV\otimes \bW) \cong i^*\bV\otimes i^*\bW$
for any $\bV, \bW\in \VB(\UConf(M))$, and 
$i_*(V\otimes W) \cong i_*V\otimes i_*W$ for any $V, W\in \VB(M)$. 

\item However, there is no tensor product in $\VB(M)$ which makes the pullback $i^*$ and the pushforward $i_*$ monoidal with respect to the Cauchy tensor product $\boxtimes$ in $\VB(\UConf(M))$: The vector bundle $i^*(\bV\boxtimes \bW) = \bV_1\otimes \bW_0 \oplus \bV_0\otimes \bW_1 \to M$ cannot be constructed from the two bundles $i^*\bV = \bV_1\to M$ and $i^*\bW = \bW_1\to M$. Similarly, $i_*V \boxtimes i_*W$ does not lie in the image of $i_*$ unless $V$ or $W$ is trivial.
\end{itemize}
\end{remark}


\section{Algebra bundles over configuration spaces}
\label{sec:3}
In this section we discuss algebra bundles in the 2-monoidal category $(\VB(\UConf(M)),\boxtimes,\otimes)$. In particular, in Subsection \ref{subsec:Cauchy algebra bundles} we construct the tensor algebra $\bT^\boxtimes$ and free commutative algebra $\bS^{\boxtimes}$ and show their properties in relation to the exterior tensor product $\exttens$ on powers of $M$. As an important example of a $\boxtimes$-algebra, we discuss the bundle of Densities over $\UConf(M)$ in Subsection \ref{subsec:density}. Finally, in Subsection \ref{subsec:Cauchy-Hadamard 2-algebra bundles}, we discuss 2-algebras and the construction of $\bS^\boxtimes S^\otimes V$, which will be instrumental in the subsequent section.

\subsection{Hadamard algebra bundles}
\label{subsec:Hadamard algebra bundles}

Following the notations of \cite[Definition 1.9]{Aguiar-Mahajan-2010}, let us call \emph{$\otimes$-algebra bundle} over $\UConf(M)$ a monoid in the category $(\VB(\UConf(M)),\otimes)$, that is, a vector bundle $\bA\in \VB(\UConf(M))$ endowed with two bundle maps 
\begin{align}\label{def:Hadamard algebra bundle}
\begin{array}{rl}
\text{$\otimes$-multiplication} & m_{\otimes}:\bA \otimes \bA \to \bA \\ 
\text{$\otimes$-unit} & u_{\otimes}:\bI_{\otimes} \to \bA 
\end{array}
\end{align}
such that $m_{\otimes}$ is associative with unit $u_{\otimes}$, in the usual sense. 
This means that there is a structure of associative unital algebra on the fibre $\bA_{\underline{x}}$ above any $k$-point configuration $\underline{x}$. 
The $\otimes$-algebra bundle $\bA$ is said to be \emph{commutative} if $m_{\otimes}\circ \tau_\otimes = m_\otimes$, where $\tau_\otimes$ is the braiding of Lemma \ref{lem:monoidal}. If $\bA$ is a graded vector algebra bundle, it is called \emph{graded-commutative} if $m_{\otimes}\circ \tau_\otimes^\sgn = m_\otimes$, where $\tau_\otimes^\sgn$ is the braiding of Remark \ref{rem:graded braiding}. 

A \emph{morphism of $\otimes$-algebra bundles} is a bundle map $\bA\to \bA'$ which respects to the multiplications and the units. 
\begin{example} \label{ex: H-unit = H-algebra}
The unit bundle $\bI_{\otimes}$ is a commutative $\otimes$-algebra bundle over $\UConf(M)$, with the pointwise multiplication $\K_{\underline {x}}\otimes \K_{\underline {x}}\to \K_{\underline x}$ over each configuration $\underline x$ and the identity $\bI_{\otimes}\to \bI_{\otimes}$ as a unit.
\end{example}

\begin{example} \label{ex:Hadamard tensor bundle}
The usual tensor powers of vector bundles can be extended from $\VB(M)$ to $\VB(\UConf(M))$, with the same universal property. 
For a vector bundle $\bV$ over $\UConf(M)$, we call \emph{Hadamard tensor bundle} of $\bV$ the vector bundle 
\begin{align*}
 \bT^{\otimes}(\bV) := \bigoplus_{n=0}^\infty \bV^{\otimes n} = I_{\otimes} \oplus \bV \oplus \bV^{\otimes 2} \oplus \cdots 
\end{align*}
with $k$-component and fibre above a configuration $\underline x$ assembled from those of each $n$-fold Hadamard tensor $\bV^{\otimes n}$, that is, 
\begin{align*}
\bT^{\otimes}(\bV)_k & = T(\bV_k), \\ 
\bT^{\otimes}(\bV)_{\underline x} &= T(\bV_{\underline x}). 
\end{align*}
Then $\bT^{\otimes}(\bV)$ is a graded $\otimes$-algebra bundle on $\UConf(M)$ with multiplication given by the usual concatenation on each fibre, denoted $\otimes$, and the unit $u:\bI_\otimes\to \bT^{\otimes}(\bV)$ given by the inclusion. 
It is easy to see that $\bT^{\otimes}(\bV)$ is the free $\otimes$-algebra in $\VB(\UConf(M))$ generated by $\bV$. 

As usual, the symmetric group $S_n$ acts on $\bV^{\otimes n}$ by permuting the $n$ factors. 
The quotients by this action give the \emph{Hadamard symmetric bundle} of $\bV$
\begin{align*}
 \bS^{\otimes}(\bV) &:= \bigoplus_{n=0}^\infty \bV^{\otimes n}/S_n 
 \cong \bigsqcup_{k} S(\bV_k) 
 \cong \bigsqcup_{\underline{x}} S(\bV_{\underline{x}}), 
\end{align*}
which inherits the structure of a $\otimes$-algebra bundle with commutative multiplication denoted $\odot$. 

If, instead, we consider the action $\sgn_n$ of $S_n$ on $\bV^{\otimes n}$ given by the signed permutation of the $n$ factors, the quotients give the \emph{Hadamard exterior bundle}
\begin{align*}
 \bL^{\otimes}(\bV) &:= \bigoplus_{n=0}^\infty \bV^{\otimes n}/\sgn_n
 \cong \bigsqcup_{k} \Lambda(\bV_k) 
 \cong \bigsqcup_{\underline{x}} \Lambda(\bV_{\underline{x}}), 
\end{align*}
which is a $\otimes$-algebra bundle with graded-commutative multiplication denoted $\owedge$. 

For $V\in \VB(M)$, the extended bundle $i_*V\in \VB(\UConf(M))$ is connected and we can apply the functors $\bT^\otimes$ and $\bS^\otimes$. 
From Remark \ref{rk:i* monoidal Hadamard}, it then follows that 
\begin{align} \label{eq:Si*=i*S}
\bT^{\otimes}(i_*V) \cong i_*T(V), 
\qquad \bS^{\otimes}(i_*V) \cong i_*S(V)
\qquad\mbox{and}\qquad 
\bL^{\otimes}(i_*V) \cong i_*\Lambda(V). 
\end{align}
\end{example}

\begin{remark} \label{rem:Hadamard Sigma}
The whole construction of this section can be dualized: one can consider \emph{$\otimes$-coalgebra bundles} over $\UConf(M)$ as comonoids in $(\VB(\UConf(M)),\otimes)$, with $\otimes$-comultiplication and $\otimes$-counit duals to the Equation \eqref{def:Hadamard algebra bundle}. 
If $\bC$ is a $\otimes$-coalgebra bundle, with $\otimes$-comultiplication $\Delta:\bC\to \bC\otimes \bC$ and counit $\varepsilon: \bC\to I_\otimes$, then its dual vector bundle $\bC^*=\Hom(\bC,\bI_\otimes)$ is naturally a $\otimes$-algebra bundle, with $\otimes$-multiplication of two bundle maps $a,b:\bC\to \bI_\otimes$ given by the convolution product $a \cdot b = m \circ (a \otimes b) \circ \delta$~:
$$
\begin{tikzcd}
\bC\arrow[r, "a \cdot b"] \arrow[d, "\Delta"'] & \bI_\otimes \\
\bC \otimes \bC \arrow[r, "a \otimes b"'] & \bI_\otimes \otimes \bI_\otimes \arrow[u, "m"']
\end{tikzcd}
$$
where $m$ is the commutative multiplication of Example \ref{ex: H-unit = H-algebra}, and unit $u(1_{\underline{x}}) = \varepsilon_{\underline{x}}$. 
Evidently, if $\bC$ is a cocommutative $\otimes$-coalgebra bundle (in the usual sense, with braiding $\tau_\otimes$ given in Lemma \ref{lem:monoidal}), then $\bC^*$ is a commutative $\otimes$-algebra bundle. 

It is easy to prove that the Hadamard tensor bundle of a connected vector bundle $\bV$, given in Example \ref{ex:Hadamard tensor bundle}, is naturally also a coalgebra bundle with the usual cofree deconcatenation comultiplication
\begin{align} \label{eq:Hadamard coproduct}
 \Delta(v_1\otimes\cdots \otimes v_n) = \sum_{i=0}^n (v_0 \otimes\cdots\otimes v_i) \,\otimes\, (v_{i+1}\otimes\cdots \otimes v_{n+1}) 
\end{align}
where $v_i\in \bV_{\underline{x}}$ and we set $v_0=v_{n+1}=1_{\underline{x}}$, and counit $\varepsilon(1_{\underline{x}})=1_{\underline{x}}$ and $\varepsilon(v_1\otimes\cdots \otimes v_n)=0_{\underline{x}}$ if $n>0$. 
We denote this $\otimes$-coalgebra bundle by $\bT^{c,\otimes}(\bV)$, to be distinguished from the previous $\otimes$-algebra bundle $\bT^\otimes(\bV) = \bT^{a,\otimes}(\bV)$. Then, the \emph{subbundle of symmetric tensors} (i.e. invariant under the action of the symmetric groups by permutations)
\begin{align} \label{eq:def Hadamard Sigma}
 \bSigma^{\otimes}(\bV) := \bigoplus_{n=0}^\infty (\bV^{\otimes n})^{S_n} \cong \bigsqcup_k \Sigma(\bV_k) \cong \bigsqcup_{\underline{x}} \Sigma(\bV_{\underline{x}})
\end{align}
and the \emph{subbundle of graded-symmetric or alternating tensors} (i.e. invariant under the action of the symmetric group by signed permutations) 
\begin{align} \label{eq:def Hadamard Lambda}
 \bL^{c,\otimes}(\bV) := \bigoplus_{n=0}^\infty (\bV^{\otimes n})^{\sgn_n} \cong \bigsqcup_k \Lambda^c(\bV_k) \cong \bigsqcup_{\underline{x}} \Lambda^c(\bV_{\underline{x}})
\end{align}
are sub-coalgebra bundles of $\bT^{c,\otimes}(\bV)$ (and \emph{not} sub-algebra bundles of $\bT^\otimes(\bV))$. 

Finally, it is easy to prove also that the linear duality of bundles intertwines the Hadamard tensor bundle with itself, bringing the $\otimes$-coalgebra structure into the $\otimes$-algebra one, and the bundle of symmetric tensors with the symmetric bundle. The only subtle point is that under duality the ind-finite and the pro-finite constructions are exchanged:
there are isomorphisms of $\otimes$-algebra bundles 
\begin{align} \nonumber
 \left(\hat{\bT}^{c,\otimes}(\bV)\right)^* & \cong \bT^{a,\otimes}(\bV^*) & \left(\bT^{c,\otimes}(\bV)\right)^* & \cong \hat{\bT}^{a,\otimes}(\bV^*) 
 \\ \label{eq:Hadamard Sigma(V)^*=S(V^*)}
 \left(\hat{\bSigma}^{\otimes}(\bV)\right)^* & \cong \bS^{\otimes}(\bV^*) & \left(\bSigma^{\otimes}(\bV)\right)^* & \cong \hat{\bS}^{\otimes}(\bV^*)
 \\ \nonumber
 \left(\hat{\bL}^{c,\otimes}(\bV)\right)^* & \cong \bL^{\otimes}(\bV^*) & \left(\bL^{c,\otimes}(\bV)\right)^* & \cong \hat{\bL}^{\otimes}(\bV^*), 
\end{align}
where the notations $\hat{\bT}^\otimes$, $\hat{\bS}^\otimes$, $\hat{\bL}$ and $\hat{\bSigma}^\otimes$, $\hat{\bL}^c$ denote the \emph{completed vector bundles}, that is, the fibrewise pro-finite collections obtained from the ind-finite original bundles, with simultaneous projection over all components of $\UConf(M)$. 
In particular, for $V\in \VB(M)$ one gets $\bSigma^{\otimes}(i_*V) \cong i_*\Sigma(V)$, and the usual isomorphism $(\hat{\Sigma}(V))^*\cong S(V)$ of algebra bundles on $M$ extends to $\UConf(M)$ using \eqref{eq:Hadamard Sigma(V)^*=S(V^*)} together with \eqref{eq:Si*=i*S}. 
\end{remark}


\subsection{Cauchy algebra bundles}
\label{subsec:Cauchy algebra bundles}

Let us now consider algebra bundles with respect to the Cauchy tensor product. Call \emph{$\boxtimes$-algebra bundle} over $\UConf(M)$ a monoid in the category $(\VB(\UConf(M)),\boxtimes)$, that is, a vector bundle $\bA\in \VB(\UConf(M))$ endowed with two bundle maps 
\begin{align}\label{def:Cauchy algebra bundle}
\begin{array}{rl}
\text{$\boxtimes$-multiplication} & m_{\boxtimes}:\bA \boxtimes \bA \to \bA \\ 
\text{$\boxtimes$-unit} & u_{\boxtimes}:\bI_{\boxtimes} \to \bA 
\end{array}
\end{align}
such that $m_{\boxtimes}$ is associative with unit $u_{\boxtimes}$. 
This means that for any $\underline{x}\in \UConf(M)$ there is necessarily a map of the form 
\begin{align*}
 m_{\underline{x}}=\sum_{\underline{x}=\underline{x}'\sqcup \underline{x}''} m_{\underline{x}',\underline{x}''} : (\bA\boxtimes \bA)_{\underline{x}} = \bigoplus_{\underline{x}=\underline{x}'\sqcup \underline{x}''} \bA_{\underline{x}'} \otimes \bA_{\underline{x}''} \to \bA_{\underline{x}}, 
\end{align*}
where $m_{\underline{x}',\underline{x}''}:\bA_{\underline{x}'}\otimes \bA_{\underline{x}''}\to \bA_{\underline{x}}$, is an associative product. Associativity means that the following diagram commutes:
$$
\begin{tikzcd}
\ds \bigoplus_{\underline{x}=\underline{x}'\sqcup \underline{x}''\sqcup \underline{x}'''} \bA_{\underline{x}'} \otimes \bA_{\underline{x}''} \otimes \bA_{\underline{x}'''}
\arrow[rr, "\sum \id\otimes m_{\underline{x}'',\underline{x}'''}"] 
\arrow[dd, "\sum m_{\underline{x}',\underline{x}''}\otimes \id"']
&& \ds \bigoplus_{\underline{x}=\underline{x}'\sqcup \underline{x}''\sqcup \underline{x}'''} \bA_{\underline{x}'} \otimes \bA_{\underline{x}''\sqcup \underline{x}'''} 
\arrow[dd, "\sum m_{\underline{x}',\underline{x}''\sqcup \underline{x}'''}"] 
\\ && \\ 
\ds \bigoplus_{\underline{x}=\underline{x}'\sqcup \underline{x}''\sqcup \underline{x}'''} \bA_{\underline{x}'\sqcup \underline{x}''} \otimes \bA_{\underline{x}'''} 
\arrow[rr, "\sum m_{\underline{x}'\sqcup \underline{x}'', \underline{x}'''}"] 
&& \ds \bA_{\underline{x}}
\end{tikzcd}
$$
If for any \emph{disjoint} configurations $\underline{x}'$ and $\underline{x}''$ we denote $a'_{\underline{x}'}\bullet a''_{\underline{x}''}= m_{\underline{x}',\underline{x}''}(a'_{\underline{x}'},a''_{\underline{x}''})$, we can shorten the associativity as
\begin{align*}
 \sum (a'_{\underline{x}'}\bullet a''_{\underline{x}''}) \bullet a'''_{\underline{x}'''} &= \sum a'_{\underline{x}'}\bullet (a''_{\underline{x}''} \bullet a'''_{\underline{x}'''})
\end{align*}
where the sum is over all splits $\underline{x}=\underline{x}'\sqcup \underline{x}''\sqcup \underline{x}'''$. 
The unit map is then a map $u_{\boxtimes}:\K_{\emptyset}\to \bA_{\emptyset}$, denoted $u_{\boxtimes}(1)=1_{\emptyset}$, such that 
\begin{align*}
 a_{\underline{x}} \bullet 1_{\emptyset} &= 1_{\emptyset} \bullet a_{\underline{x}} = a_{\underline{x}}. 
\end{align*}

The $\boxtimes$-algebra bundle $\bA$ is said to be \emph{commutative} if $m_{\boxtimes}\circ \tau_\boxtimes = m_\boxtimes$, where $\tau_\boxtimes$ is the braiding of Lemma \ref{lem:monoidal}. 
This means that if a configuration $\underline{x}$ splits as $\underline{x}'\sqcup \underline{x}''$, and therefore also as $\underline{x}''\sqcup \underline{x}'$, the sums over the splits of the two multiplications 
$\bA_{\underline{x}'}\otimes \bA_{\underline{x}''}\to \bA_{\underline{x}}$ and $\bA_{\underline{x}''}\otimes \bA_{\underline{x}'}\to \bA_{\underline{x}}$ give the same result, that is, 
\begin{align*}
 \sum a'_{\underline{x}'}\bullet a''_{\underline{x}''} = \sum a''_{\underline{x}''}\bullet a'_{\underline{x}'}. 
\end{align*}
If $\bA$ is a graded $\boxtimes$-algebra bundle (that is, $\bA$ is a graded vector bundle and the $\boxtimes$-multiplication preserves the grading), then $\bA$ is called \emph{graded-commutative} if $m_\boxtimes \circ \tau_\boxtimes^\sgn = m_\boxtimes$, where $\tau_\boxtimes^\sgn$ is the braiding of Remark \ref{rem:graded braiding}. 

A \emph{morphism of $\boxtimes$-algebra bundles} is a bundle map $\bA\to \bA'$ which respect to the multiplications and the units. 

\begin{example} \label{ex: H-unit = C-algebra}
The unit bundle $\bI_{\otimes}$ is a commutative $\boxtimes$-algebra bundle over $\UConf(M)$, with multiplication $\mu$ and unit $\nu$ given in Theorem \ref{thm: VB(Conf(M)) is 2-monoidal}. 
\end{example}

We now give another example of a commutative $\boxtimes$-algebra bundle over $\UConf(M)$, which plays a key role in field theory. 
We start by iterating the construction of Lemma \ref{eq:pullback pushforward} on multiple copies of the same vector bundle, to define the $\boxtimes$-tensor algebra bundle, and then consider the coinvariants under the action of the symmetric group to define the symmetric $\boxtimes$-tensor algebra. 

\begin{definition} \label{def:Cauchy tensor bundle}
Let $V\in\VB(M)$. We call \emph{Cauchy tensor bundle} of $V$ the vector bundle $\bT^{\boxtimes}(V)\in \VB(\UConf(M))$ given by the iterated $\boxtimes$-tensor powers of $i_*V$
\begin{align*}
\bT^{\boxtimes}(V) := \bigoplus_{n=0}^\infty (i_*V)^{\boxtimes n} 
= I_{\boxtimes} \oplus i_*V \oplus (i_*V)^{\boxtimes 2} \oplus \cdots
\end{align*}
with $k$-component coinciding with the $k$-fold Cauchy tensor $(i_*V)^{\boxtimes k}$, that is 
\begin{align*}
\bT^{\boxtimes}(V)_k \equiv \bT^{\boxtimes k}(V) & =
\begin{cases}
 \K \times \UConf_0(M) & \text{if } k=0 \\
 \underbrace{(i_*V\boxtimes \cdots \boxtimes i_*V)}_{\text{$k$ times}}|_{\UConf_k(M)} & \text{if } k\geq 1
\end{cases} 
\end{align*}
\end{definition}

On a configuration $\underline x=\{x_1,...,x_k\}$, the fibre of $\bT^{\boxtimes}(V)$ is the tensor product of fibres of $V$ over single points $x_1,...,x_k\in M$, up to a permutation of the $k$ points, that is, 
\begin{align}\label{eq:TVoverapoint}
\bT^{\boxtimes}(V)_{\underline x} = \bT^{\boxtimes k}(V)_{\underline x} 
&= \bigoplus_{\sigma\in S_k} V_{\sigma(x_1)}\otimes \cdots \otimes V_{\sigma(x_k)} . 
\end{align}
If, for arbitrary $k$ vector spaces $W_1,...,W_k$ we denote 
\begin{align}\label{eq:boxtimes on vector spaces}
 W_1\boxtimes \cdots \boxtimes W_k &:= \bigoplus_{\sigma\in S_k} W_{\sigma(1)}\otimes \cdots \otimes W_{\sigma(k)}, 
\end{align}
we can shortly write $\bT^{\boxtimes k}(V)_{\underline x} = V_{x_1} \boxtimes \cdots \boxtimes V_{x_k}$. 
Note that, with this convention, when we write $v_1\otimes\cdots\otimes v_k \in V_{x_1} \boxtimes \cdots \boxtimes V_{x_k}$ we mean that there is a permutation $\sigma\in S_k$ such that $v_i\in V_{x_{\sigma^{-1}(i)}}$ for any $i=1,\dots,k$. 

\begin{theorem}\label{thm:Cauchy tensor bundle}
The Cauchy tensor bundle $\bT^{\boxtimes}(V)$ is a graded $\boxtimes$-algebra bundle over $\UConf(M)$, with multiplication 
$$
m:\bT^{\boxtimes}(V)\boxtimes \bT^{\boxtimes}(V)\to \bT^{\boxtimes}(V) 
$$
given by the concatenation 
$$
m(v_1\otimes\cdots\otimes v_i,w_1\otimes\cdots\otimes w_j) = v_1\otimes\cdots\otimes v_i \otimes w_1\otimes\cdots\otimes w_j, 
$$
and unit $u:\bI_\boxtimes\to \bT^{\boxtimes}(V)$ given by the inclusion. 
Moreover, $\bT^{\boxtimes}(V)$ is the free $\boxtimes$-algebra in $\VB(\UConf(M))$ generated by $V\in \VB(M)$. 
\end{theorem}

\begin{proof}
To show that $\bT^{\boxtimes}(V)$ is a unital associative algebra, it is enough to show that $m$ is  well-defined (the rest follows immediately). 
For a $k$-point configuration $\underline{x}$, we have 
\begin{align}\label{eq:TboxT}
\left( \bT^{\boxtimes}(V)\boxtimes \bT^{\boxtimes}(V)\right)\!_{\underline x} 
& = \bigoplus_{\underline x =\underline x'\sqcup \underline x''} \bT^{\boxtimes}(V)_{\underline x'}\otimes \bT^{\boxtimes}(V)_{\underline x''}. 
\end{align}
Notice, in particular, that the two splits $\underline{x}=\underline{x}\sqcup \emptyset$ and $\underline{x}=\emptyset \sqcup \underline{x}$ are included in the sum. Call $i=|\underline{x}'|=0,\dots,k$, so that $j=|\underline{x}''|=k-i$, and consider the set of $(i,j)$-shuffle permutations $\USh(i,j)$. 
If $i=0$ or $j=0$, the only $(0,k)$ and $(k,0)$ unshuffles are the unit permutation, and the corresponding terms in eq. \eqref{eq:TboxT} are two copies of $\bT^{\boxtimes k}(V)_{\underline x}$.
For $\underline{x}=\{x_1,\dots,x_k\}$, the sum in eq. \eqref{eq:TboxT} can be written as
\begin{align*}
& \underset{\sigma\in \USh(i,j)}{\underset{i,j\geq 0}{\bigoplus_{i+j=k}}} \bT^{\boxtimes i}(V)_{\{x_{\sigma^{-1}(1)},\dots,x_{\sigma^{-1}(i)}\}}\otimes \bT^{\boxtimes j}(V)_{\{x_{\sigma^{-1}(i+1)},\dots,x_{\sigma^{-1}(k)}\}} \\
&= \underset{\sigma\in \USh(i,j)}{\underset{i,j\geq 0}{\bigoplus_{i+j=k}}} \!\! \left(\bigoplus_{\tau_1\in S_i}V_{x_{\tau^{-1}(\sigma^{-1}(1))}}\otimes \cdots\otimes V_{x_{\tau^{-1}(\sigma^{-1}(i))}}\right) \otimes \left( \bigoplus_{\tau_2^\in S_j}V_{x_{\tau_2^{-1}(\sigma^{-1}(i+1))}}\otimes \cdots \otimes V_{x_{\tau_2^{-1}(\sigma^{-1}(k))}}\right).
\end{align*}
For each $i=0,...,k$ and $j=k-i$, there are $\binom{k}{i}=\frac{k!}{i!j!}$ $(i,j)$-unshuffles. Then, the sum over all $(i,j)$-unshuffles, combined with the sum over all permutations on $i$ and on $j$ elements, produces 
$$
\sum_{i=0}^k \binom{k}{i} i! j! = \sum_{i=0}^k k! = (k+1)k!
$$
terms, corresponding to $k+1$ sums over all permutations on $k$ elements, and we finally have
\begin{align} \nonumber
\left(\bT^{\boxtimes}(V)\boxtimes \bT^{\boxtimes}(V)\right)_{\underline x} 
&= \bigoplus_{i=1}^{k+1}\, \bigoplus_{\tilde{\sigma}\in S_k} V_{x_{\tilde{\sigma}^{-1}(1)}}\otimes \cdots \otimes V_{x_{\tilde{\sigma}^{-1}(k)}} 
\\ \label{eq:Tbox boxtimes Tbox = Tbox}
&= \bigoplus_{i=1}^{k+1}\, \bT^{\boxtimes k}(V)_{\underline{x}}. 
\end{align}
The multiplication $m$ is then just the natural projection from each copy of $\bT^{\boxtimes k}(V)_{\underline x}$ to itself. 

This implies that, for any $k\geq 2$, any $k$-fold tensor $v=v_1\otimes\cdots\otimes v_k\in \bT^{\boxtimes k}(V)_{x_1,...,x_k}$ can be seen as the product of its $1$-fold components $v_i\in V_{x_i}$. Adopting a standard abuse of notations, we then denote by $\boxtimes$ the concatenation $m$ on $\bT^{\boxtimes}$, and write
\begin{align}\label{eq:concatenation}
 m(v,w) &= v \boxtimes w 
 \\ \nonumber 
 v_1 \otimes\cdots\otimes v_k &= v_1 \boxtimes\cdots\boxtimes v_k. 
\end{align}

To see that $\bT^{\boxtimes}(V)$ is a free $\boxtimes$-algebra in $\VB(\UConf(M))$ generated by $V$, it suffices to note that $\bT^{\boxtimes}(V)$ comes with a natural inclusion $i_*V \hookrightarrow \bT^{\boxtimes}(V)$.

Then, for any $\boxtimes$-algebra bundle $\bA$ over $\UConf(M)$, with multiplication denoted by $m_{\bA}(a,a')=a\bullet_{\bA} a'$, and for any vector bundle morphism $f:i_*V\to \bA$ over $\UConf(M)$, there exists a unique algebra bundle morphism $\tilde{f}: \bT^{\boxtimes}(V) \to \bA$ which makes the usual diagram 
$$
\begin{tikzcd}
i_*V \arrow[r, hook] \arrow[dr, "f"']
& \bT^{\boxtimes}(V) \arrow[d, "\tilde{f}"] \\
& \bA
\end{tikzcd}
$$
commutative. For any $v_1\boxtimes\cdots\boxtimes v_k\in \bT^{\boxtimes k}(V)$, it is given by 
\begin{align*}
 \tilde{f}(v_1\boxtimes\cdots\boxtimes v_k) &= 
 \begin{cases}
 u_{\bA}(1) & \text{if } k=0 \\ 
 f(v_1)\bullet_{\bA} f(v_2)\bullet_{\bA} \cdots \bullet_{\bA} f(v_k) & \text{if } k\geq 1. 
 \end{cases}
\end{align*}
\end{proof}

A symmetric and a graded-symmetric version of $\bT^{\boxtimes}$ can be constructed by identifying the different summands of \eqref{eq:TVoverapoint}, which are all isomorphic. 
For any $n\geq 1$, consider the canonical action of the symmetric group $S_n$ on $\bT^{\boxtimes n}(V)=(i_*V)^{\boxtimes n}$ given, for any $\sigma\in S_n$ and any $v = v_{1}\otimes v_{2}\otimes\cdots\otimes v_{n} \in \bT^{\boxtimes n}(V)$, by
$$
\sigma(v) = v_{\sigma(1)}\otimes v_{\sigma({2})}\otimes\cdots\otimes v_{\sigma({n})}, 
$$
and the signature action $\sgn(\sigma)$ given by 
$$
\sgn(\sigma)(v) = (-1)^{|\sigma|} v_{\sigma(1)}\otimes v_{\sigma({2})}\otimes\cdots\otimes v_{\sigma({n})}. 
$$

\begin{definition}
For $V\in \VB(M)$, we call \emph{Cauchy symmetric bundle} $\bS^{\boxtimes}(V)$ the fibrewise quotient of $\bT^{\boxtimes}(V)$ by the action of the appropriate symmetric group, that is, the bundle 
\begin{align*}
\bS^{\boxtimes}(V) := \bigoplus_{n=0}^\infty (i_*V)^{\boxtimes n}/S_n 
= I_{\boxtimes} \oplus i_*V \oplus (i_*V)^{\boxtimes 2}/S_2 \oplus \cdots
\end{align*}
with $k$-component
\begin{align*}
\bS^{\boxtimes}(V)_k \equiv \bS^{\boxtimes k}(V) := \bT^{\boxtimes k} (V)/S_k. 
\end{align*}
Similarly, we call \emph{Cauchy exterior bundle} $\bL^{\boxtimes}(V)$ the fibrewise quotient of $\bT^{\boxtimes}(V)$ by the signature action, that is, the bundle 
\begin{align*}
\bL^{\boxtimes}(V) := \bigoplus_{n=0}^\infty (i_*V)^{\boxtimes n}/\sgn_n 
= I_{\boxtimes} \oplus i_*V \oplus (i_*V)^{\boxtimes 2}/\sgn_2 \oplus \cdots
\end{align*}
with $k$-component
\begin{align*}
\bL^{\boxtimes}(V)_k \equiv \bL^{\boxtimes k}(V) := \bT^{\boxtimes k} (V)/\sgn_k. 
\end{align*}
\end{definition}

Explicitly, on a configuration $\underline x\in \UConf_k(M)$, the fibre of $\bS^{\boxtimes k}(V)$ and of $\bL^{\boxtimes k}(V)$ is the set of equivalence classes $[v]$ of Cauchy tensors $v\in \bT^{\boxtimes k}(V)_{\underline x}$ under the equivalence relation generated by $v \sim \sigma(v)$ and by $v \sim (-1)^{\sgn(\sigma)} \sigma(v)$ for any $\sigma \in S_k$, namely
\begin{align*}
\bS^{\boxtimes k}(V)_{\underline x} & = {\bT^{\boxtimes k}(V)_{\underline x}} \ \big{/} \langle v-\sigma(v)~|~\sigma \in S_k\rangle \\ 
\bL^{\boxtimes k}(V)_{\underline x} & = {\bT^{\boxtimes k}(V)_{\underline x}} \ \big{/} \langle v-(-1)^{\sgn(\sigma)}\sigma(v)~|~\sigma \in S_k\rangle. 
\end{align*}

\begin{theorem}\label{thm:Cauchy symmetric bundle}
The Cauchy symmetric bundle $\bS^{\boxtimes}(V)$ is a commutative $\boxtimes$-algebra bundle over $\UConf(M)$, with multiplication and unit induced by those of $\bT^{\boxtimes}(V)$. 
Moreover, $\bS^{\boxtimes}(V)$ is the free commutative $\boxtimes$-algebra in $\VB(\UConf(M))$ generated by $V\in \VB(M)$. 

Similarly, the Cauchy exterior bundle $\bL^{\boxtimes}(V)$ is a graded-commutative $\boxtimes$-algebra bundle over $\UConf(M)$, and it is the free graded-commutative $\boxtimes$-algebra in $\VB(\UConf(M))$ generated by $V\in \VB(M)$. 
\end{theorem} 

\begin{proof}
To see that the quotient space $\bS^{\boxtimes}(V)$ is indeed a smooth vector bundle over $\UConf(M)$, one can proceed by the same strategy as in the proof of Theorem \ref{thm:Cauchy vs external}, i.e. pull back the problem to a small neighborhood in $\OConf(M)$ where the involved vector bundles are trivial.
The multiplication $m$ and the unit $u$ of $\bT^{\boxtimes}(V)$ are graded, and pass to the quotient by the graded action of $S_k$. The induced multiplication $\bar m$ on $\bS^{\boxtimes}(V)$ is clearly commutative, and therefore the Cauchy symmetric space is indeed a commutative $\boxtimes$-algebra bundle on $\UConf(M)$. The fact that it is a free commutative $\boxtimes$-algebra bundle generated by $V\in \VB(M)$ follows from the fact that $\bT^{\boxtimes}(V)$ is the free $\boxtimes$-algebra.

The same arguments apply to the Cauchy exterior bundle $\bL^\boxtimes(V)$.
\end{proof}
\medskip 

To simplify the notations, an equivalence class $[v]\in \bS^{\boxtimes k}(V)$ is denoted by a representative $k$-fold tensor $v\in \bT^{\boxtimes k}(V)$. To avoid ambiguity, in analogy with the notation adopted for $\bT^{\boxtimes}(V)$ in the Equation \eqref{eq:concatenation}, let us introduce a notation $\boxdot$ for the product $\boxtimes$ factored out by the symmetric group action, and therefore write
\begin{align*}
 \bar m([v],[w]) &= v \boxdot w \\ 
 [v_1 \boxtimes\cdots\boxtimes v_k] &= v_1 \boxdot\cdots\boxdot v_k, 
\end{align*}
and finally also 
\begin{align*}
\bS^{\boxtimes}(V)_k \equiv \bS^{\boxtimes k}(V) & :=
\begin{cases}
 \K \times \UConf_0(M) & \text{if } k=0 \\
 \underbrace{i_*V\boxdot \cdots \boxdot i_*V}_{\text{k times}} & \text{if } k\geq 1
\end{cases} 
\end{align*}
If, applying eq. \eqref{eq:boxtimes on vector spaces}, for $k$ vector spaces $W_1,...,W_k$ we denote 
\begin{align} \label{eq:boxdot on vector spaces}
 W_1\boxdot \cdots \boxdot W_k &:= (W_1\boxtimes \cdots \boxtimes W_k)/S_k, 
\end{align}
we can also shortly write $S^{\boxtimes k}(V)_{\underline{x}} = V_{x_1}\boxdot\cdots \boxdot V_{x_k}$. 

Exactly in the same way, we can denote by $\boxwedge$ the product $\boxtimes$ factored out by the signature action of the symmetric groups and write $\bL^{\boxtimes k}(V)_{\underline{x}}=V_{x_1}\boxwedge \cdots \boxwedge V_{x_k}$. 
\\

Finally, each fibre $\bT^{\boxtimes k}(V)_{\underline{x}} = V_{x_1}\boxtimes \cdots \boxtimes V_{x_k}$ contains $k!$ isomorphic copies of the same vector space $V_{x_1}\otimes \cdots \otimes V_{x_k}$, which differ in the order of the base points. 
Instead, the fibre $\bS^{\boxtimes k}(V)_{\underline{x}}= V_{x_1}\boxdot \cdots \boxdot V_{x_k}$ is exactly isomorphic to $V_{x_1}\otimes \cdots \otimes V_{x_k}$, even though it is a \emph{symmetric space} (with respect to the Cauchy monoidal structure $\boxtimes$) and the order of the base points does not matter. 
This is the key result which allows us to use the symmetric Cauchy tensor algebra bundle over $\UConf(M)$ to represent the commutative product of multilocal functionals in field theory. In fact, this remark can be made more precise and gives a direct comparison between the external vector bundle $V\exttens V\to M\times M$ outside the diagonal (cf. Lemma \ref{lem:external tensor product on OConf(M)}) and the symmetric Cauchy vector bundle $V\boxdot V\to \UConf_2(M)$. 
The extension to the diagonal requires orbifolds is addressed in the following article \cite{nguyen2026}.

\begin{proposition} \label{prop:comparison external product with Cauchy product}
For any vector bundle $V\to M$ and for any $k\geq 0$ the vector bundles $V^{\exttens k}\to \OConf_k(M)$ and $q^*_k\bS^{\boxtimes}(V)\to \OConf_k(M)$ are isomorphic.
\end{proposition}

\begin{remark} The above proposition could also be taken as the definition of $\bS^{\boxtimes}(V)$, as was done in \cite[Definition 2.3]{tevelev-torres-2023} in the context of algebraic geometry. One interpretation of the $\boxtimes$-tensor product is: It is the tensor product with respect to which $\bS^{\boxtimes}(V)$ is the symmetric algebra object.
\end{remark}

\begin{proof}
The natural projection $p:\bT^{\boxtimes}(V)\to \bS^{\boxtimes}(V)$ induces a projection $q^*_kp:q_k^*\bT^{\boxtimes}(V)\to q_k^*\bS^{\boxtimes}(V)$. By Theorem \ref{thm:Cauchy vs external}, we have
$$
q_k^*\bT^{\boxtimes}(V)= \bigoplus_{\sigma\in S_k} \sigma^* V^{\exttens k}.
$$ 
In particular there is a natural inclusion $j:V^{\exttens k}\to q_k^*\bT^{\boxtimes}(V)$. The composition $q^*_kp\circ j$ can be seen to be an isomorphism by a fibrewise computation.
\end{proof}

The Cauchy tensors $\bT^{\boxtimes}$ and $\bS^{\boxtimes}$ are functors on $\VB(M)$ with values respectively in the category $\mathrm{Alg}(\VB(\UConf(M)),\boxtimes)$ of algebra bundles and in the category $\mathrm{CAlg}(\VB(\UConf(M)),\boxtimes)$ of commutative algebra bundles over $\UConf(M)$. So we can summarize the situation in the following (non-commutative) diagram, where the vertical arrows are forgetful functors:
$$
\begin{tikzcd}
&& \mathrm{CAlg}(\VB(\UConf(M)),\boxtimes) \arrow[d] \\
&&\mathrm{Alg}( \VB(\UConf(M)),\boxtimes)\arrow[d] \\
\VB(M) \arrow[uurr, "\bS^\boxtimes"] \arrow[urr, "\bT^{\boxtimes}"] \arrow[rr, "i_*"] && \VB(\UConf(M))
\end{tikzcd}
$$

A direct computation (using the compatibility of $i_*$ with the Hadamard tensor product, c.f. Remark \ref{rk:i* monoidal Hadamard}) shows that $\bS^{\boxtimes}$ is also compatible with the Hadamard tensor product: 

\begin{lemma}\label{lem: S^boxtimes strongly monoidal}
The functor $\bS^{\boxtimes}$ is strongly monoidal with respect to the Hadamard tensor product, in the sense of \cite{MacLane-1998}, i.e. there are natural isomorphisms 
\begin{align*}
 \bS^{\boxtimes}(V \otimes W) &\cong \bS^{\boxtimes}(V) \otimes \bS^{\boxtimes}(W) \\ 
 \bS^{\boxtimes}(I_{\otimes}) &\cong \bI_{\otimes}. 
\end{align*}
\end{lemma}

\begin{remark} \label{rem:Cauchy Sigma}
The whole construction of this section can also be dualized: one can consider \emph{$\boxtimes$-coalgebra bundles} over $\UConf(M)$ as comonoids in $(\VB(\UConf(M)),\boxtimes)$, with $\boxtimes$-comultiplication $\Delta:\bC\to \bC\boxtimes \bC$ and $\boxtimes$-counit $\varepsilon: \bC\to I_\boxtimes$ duals to eq. \eqref{def:Cauchy algebra bundle}, satisfying the dual coassociativity and counitality axioms. 
For the comultiplication of $c\in \bC$ (which necessarily lives in a finite-dimensional fibre $\bC_{\underline{x}}$), we can adopt the usual Sweedler notation 
\begin{align*}
 \Delta(c) &= \sum c_{(1)} \boxtimes c_{(2)}, 
\end{align*}
where the sum is finite, which in fact means 
\begin{align*}
 \Delta(c_{\underline{x}}) &= \sum_{\underline{x}=\underline{x}'\sqcup \underline{x}''} \sum c_{(1),\underline{x}'} \otimes c_{(2),\underline{x}''}. 
\end{align*}

If $\bC$ is a $\boxtimes$-coalgebra bundle, then its dual vector bundle $\bC^*=\Hom(\bC,\bI_\otimes)$ is naturally a $\boxtimes$-algebra bundle: the $\boxtimes$-multiplication $\bullet$ of two bundle maps $a,b:\bC\to \bI_\otimes$ is given by the convolution 
$$
\begin{tikzcd}
\bC\arrow[r, "a \bullet b"] \arrow[d, "\Delta"'] & \bI_\otimes \\
\bC \boxtimes \bC \arrow[r, "a \boxtimes b"'] & \bI_\otimes \boxtimes \bI_\otimes \arrow[u, "\mu"']
\end{tikzcd}
$$
where $\mu:\bI_\otimes \boxtimes \bI_\otimes \to \bI_\otimes$ is the commutative multiplication of Theorem \ref{thm: VB(Conf(M)) is 2-monoidal}, and the unit $u:\bI_\boxtimes \to \bC^*$ is given by $u(1_{\emptyset}) = \varepsilon_{\emptyset}$ on the zero-point configuration and equal to the zero map on any other configuration $\underline{x}$ (that is, formally $u$ is the composition of the map $\nu:\bI_\boxtimes \to \bI_\otimes$ of Theorem \ref{thm: VB(Conf(M)) is 2-monoidal} with the dual of the counit $\varepsilon$). 
Clearly, if $\bC$ is a cocommutative $\boxtimes$-coalgebra bundle (in the usual sense, with braiding $\tau_\boxtimes$ given in Lemma \ref{lem:monoidal}), then $\bC^*$ is a commutative $\boxtimes$-algebra bundle. 

If $\langle\ ,\ \rangle:\bC^*\times \bC\to \bI_\otimes$ denotes the fibrewise pairing of dual vector bundles, then the multiplication on $\bC^*$ and the comultiplication on $\bC$ are related over each configuration $\underline{x}$ by the usual identity 
\begin{align}
 \langle a\bullet b,c\rangle_{\underline{x}} &= \langle a \boxtimes b,\Delta(c) \rangle_{\underline{x}} 
 = \sum_{\underline{x}=\underline{x}'\sqcup \underline{x}''} \langle a_{\underline{x}'},c_{(1),\underline{x}'}\rangle \langle b_{\underline{x}''},c_{(2),\underline{x}''}\rangle . 
\end{align}

With a bit of work, one can prove that the Cauchy tensor bundle of a vector bundle $V$ over $M$, given in Definition \ref{def:Cauchy tensor bundle}, is naturally also a coalgebra bundle that we shall denote by $\bT^{c,\boxtimes}(V)$, with the cofree deconcatenation comultiplication $\Delta: \bT^{c,\boxtimes}(V) \to \bT^{c,\boxtimes}(V) \boxtimes \bT^{c,\boxtimes}(V)$ given on any $k$-point configuration $\underline{x}=\{x_1,\ldots,x_k\}$, and on any vector $v=v_1\otimes\cdots \otimes v_k \in \bT^{c,\boxtimes k}(V)_{\underline{x}}$, with $v_i \in V_{x_{\sigma^{-1}(i)}}$ for some permutation $\sigma\in S_k$, by 
\begin{align} \label{eq:Cauchy tensor coproduct}
 \Delta(v) &= \sum_{\underline{x}=\underline{x}'\sqcup \underline{x}''} v|_{\underline{x}'} \otimes v|_{\underline{x}''}, 
\end{align}
where the restrictions of the tensor $v$ to any split of the form $\underline{x}'=\{x_{i_1},\ldots,x_{i_{k'}}\}$ and $\underline{x}''=\{x_{j_1},\ldots,x_{j_{k''}}\}$, with $\{i_1,\ldots,i_{k'}\} \sqcup \{j_1,\ldots,j_{k''}\} = \{1,\ldots,k\}$, are 
\begin{align*}
 v|_{\underline{x}'} & = v_{\sigma^{-1}(i_1)} \otimes\cdots \otimes v_{\sigma^{-1}(i_{k'})} \in V_{x_{i_1}} \otimes \cdots \otimes V_{x_{i_{k'}}} \\ 
 v|_{\underline{x}''} & = v_{\sigma^{-1}(j_1)} \otimes\cdots \otimes v_{\sigma^{-1}(j_{k''})} \in V_{x_{j_1}} \otimes \cdots \otimes V_{x_{j_{k''}}}. 
\end{align*}
The key ingredients to prove that $\Delta$ is a  well-defined coassociative comultiplication is the equality of eq. \eqref{eq:Tbox boxtimes Tbox = Tbox} in the proof of Theorem \ref{thm:Cauchy tensor bundle}. 
The counit $\varepsilon: \bT^{c,\boxtimes}(V) \to I_\boxtimes$ is of course the identity on the zero-point component $\bT^{c,\boxtimes 0}(V)=I_\boxtimes$ and the zero map elsewhere. 

Then we can define the \emph{subbundle of symmetric tensors}, i.e. invariant under the action of the symmetric groups, as
\begin{align*}
 \bSigma^{\boxtimes}(V) := \bigoplus_{n=0}^\infty (i_*V^{\boxtimes n})^{S_n} \cong \bigsqcup_k \bT^{\boxtimes k}(V)^{S_k}, 
\end{align*}
with fibre over $\underline{x}\in \UConf(M)$ given by the symmetric Cauchy tensors
\begin{align*}
 \bSigma^{\boxtimes k}(V)_{\underline{x}} = \{ v\in \bT^{\boxtimes k}(V)\ |\ \mbox{$\sigma(v)=v$ for any $\sigma\in S_k$}\}. 
\end{align*}
For instance, for any vectors $a_x\in V_x$ and $b_y\in V_y$, the Cauchy tensor $a_x\otimes b_y + b_y \otimes a_x\in \bT^{\boxtimes 2}(V)_{\{x,y\}}$ belongs to $\bSigma^{\boxtimes 2}(V)_{\{x,y\}}$, while each separate term $a_x\otimes b_y$ and $b_y\otimes a_x$ does not. 

It is easy to see that the coproduct \eqref{eq:Cauchy tensor coproduct} is closed on symmetric tensors, therefore $\bSigma^{\boxtimes}(V)$ is a sub-coalgebra bundle of $\bT^{c,\boxtimes}(V)$. 
Similarly, it is easy to see that the concatenation product of Theorem \ref{thm:Cauchy tensor bundle} is \emph{not} closed on symmetric tensors, so $\bSigma^{\boxtimes}(V)$ is not a sub-algebra bundle of $\bT^{\boxtimes}(V)= \bT^{a,\boxtimes}(V)$. 

Finally, one can prove that the linear duality of bundles intertwines the Cauchy tensor bundle with itself, bringing the $\boxtimes$-coalgebra structure into the $\boxtimes$-algebra one, and the bundle of symmetric tensors with the symmetric bundle. The subtle point occurring for the Hadamard tensor product and mentioned in Remark \ref{rem:Hadamard Sigma}, that ind-finite bundles dualize into pro-finite ones, does not occur for the Cauchy tensor bundles constructed on a vector bundle $V\to M$, because Cauchy powers add fibres above new base points and do not affect the size of each fibre. The output is that there are isomorphisms of $\boxtimes$-algebra bundles
\begin{align*} \nonumber
 \left(\bT^{c,\boxtimes}(V)\right)^* & \cong \bT^{a,\boxtimes}(V^*) 
 \\ 
 \left(\bSigma^{\boxtimes}(V)\right)^* & \cong \bS^{\boxtimes}(V^*). 
\end{align*}
In particular, this result can be applied to the vector bundle $\Sigma^\otimes(V)\to M$ and we have canonical isomorphisms of Cauchy algebra bundles 
\begin{align} \label{eq:Cauchy Sigma Sigma(V)^*=SS(V^*)}
 & \left(\bSigma^\boxtimes(\hat{\Sigma}^\otimes(V))\right)^* \cong \bS^\boxtimes\left((\hat{\Sigma}^\otimes(V))^*\right) \cong \bS^\boxtimes(S^\otimes(V^*)) 
 \\ \nonumber
 & \left(\bSigma^\boxtimes(\Sigma^\otimes(V))\right)^* \cong \bS^\boxtimes\left((\Sigma^\otimes(V))^*\right) \cong \bS^\boxtimes(\hat{S}^\otimes(V^*)) . 
\end{align}
\end{remark}


\subsection{Density bundle}
\label{subsec:density}

Integrals over an orientable manifold are usually carried out using a volume form, a differential form in the degree of the manifold. However, when integrating over (potentially) non-orientable manifolds, one has to work with densities, since a global volume form does not exist. In this section we recall the definition of densities and show that on $\UConf(M)$ they form a commutative $\boxtimes$-algebra bundle (a Cauchy algebra bundle). 

Let us recall, see for example \cite{Grosser-Kunzinger-Oberguggenberger-Steinbauer-2001}, that a \emph{density} on a vector space $V$ of dimension $d$ is a continuous, real-valued function $\nu$ of $d$ linearly independent vectors $v=(v_1,\dots,v_d)\in V$ such that if $v'=(v'_1,\dots,v'_d)\in V$ is another set of $d$ linearly independent vectors in $V$, and $v'=Av$ with $A \in \GL(d,\R)$, then
\[
\nu(v') = |\det A| \ \nu(v).
\]
In other words, a density is a \emph{semi-linear} map $\nu: \Lambda^d V \to \R$. 
Then, densities on $V$ form a one-dimensional vector space $\Lambda^d V^\ast$ which is seen as a semi-linear representation of $\GL(d,\R)$ under the map
\[
\begin{array}{rcl} 
 GL(d,\R)& \longrightarrow & \R^+ \\
 A &\mapsto & |\det A|^{-1} . 
 \end{array}
\]

If $M$ is a smooth manifold of dimension $d$, the \emph{density bundle} of $M$ is the bundle $\Dens_M$ whose fibre at a point $x \in M$ is given by the densities on the tangent space $T_xM$. It can also be seen as the bundle $\Dens_M = \F M \times_{\GL(d,\R)} \R$ associated to the frame bundle $\F M$, with typical fibre $\R$ isomorphic to the fibres of $\Lambda^d(T^\ast M)$ and the action of $A\in \GL(d,\R)$ on $\lambda\in \R$ given by $|\det A|^{-1}\lambda$. 
A \emph{density} on $M$ is a smooth section of the density bundle.

On the manifold $\UConf(M)=\bigsqcup_{k=0}^\infty \UConf_k(M)$, the density bundle is the collection of density bundles on each $k$-component:
\begin{align} \label{def:density bundle}
 \Dens_{\UConf(M)} = \bigsqcup_{k=0}^\infty \Dens_{\UConf_k(M)}. 
\end{align}
For any $k$, the bundle $\Dens_{\UConf_k(M)}$ is a usual vector bundle of rank $1$. 
There is however another canonical vector bundle of rank $1$ which can be defined on $\UConf_k(M)$, the symmetric Cauchy $k$-tensor powers of $\Dens_M$, i.e. $\bS^{\boxtimes k}(\Dens_M)$. Luckily, these two bundles coincide.

\begin{theorem} \label{thm:Dens_Conf(M)=S(Dens_M)}
There is an isomorphism of vector bundles over $\UConf(M)$
$$
\Dens_{\UConf(M)}\cong \bS^\boxtimes(\Dens_M).
$$
Therefore $\Dens_{\UConf(M)}$ is a commutative $\boxtimes$-algebra bundle on $\UConf(M)$, with an algebra structure inherited from the free commutative $\boxdot$ product of $S^{\boxtimes}(\Dens_{M})$.
\end{theorem}
 \begin{proof} 
To prove the assertion we consider the following diagram of vector bundles
\[\begin{tikzcd}
	{\Dens_{\UConf_k(M)}} & {\Dens_{\OConf_k(M)}} & {\Dens_M^{\exttens k}} & {q^*_k\bS^{\boxtimes k}\Dens_M} & {\bS^{\boxtimes k}\Dens_M} \\
	{{\UConf(M)_k}} && {{\OConf_k(M)}} && {{\UConf_k(M)}} 
	\arrow[from=1-1, to=2-1]
	\arrow["\alpha"', from=1-2, to=1-1]
	\arrow["\beta", tail reversed, no head, from=1-2, to=1-3]
	\arrow[from=1-2, to=2-3]
	\arrow["\gamma", from=1-3, to=1-4]
	\arrow[from=1-3, to=2-3]
	\arrow["\delta", from=1-4, to=1-5]
	\arrow[from=1-4, to=2-3]
	\arrow[from=1-5, to=2-5]
	\arrow["{q_k}"', from=2-3, to=2-1]
	\arrow["{q_k}", from=2-3, to=2-5]
\end{tikzcd}\]
Let us have a look at the individual maps:
\begin{itemize}
 \item The arrow $\alpha$ is a local isomorphism of vector bundles, since $q_k$ is a local diffeomorphism, and $\alpha$ is induced by its differential.
 \item The arrow $\gamma$ is an isomorphism following Theorem \ref{thm:Cauchy vs external}.
 \item The arrow $\delta$ is a local isomorphism by the natural properties of the pullback (see Remark \ref{rem:VB on OConf(M)}).
\end{itemize}
Let us define the remaining map $\beta$ and show that it is an isomorphism of vector bundles.
Recall that the tangent bundle of $\OConf_k(M)$ is
\[
\bT\OConf_k(M) =\pi_1^*TM\oplus \pi_2^*TM\oplus \cdot \oplus \pi_k^*TM = \displaystyle{\bigsqcup_{(x_1,x_2, \cdots, x_k)}T_{x_1}M\oplus T_{x_2}M \oplus \cdots \oplus T_{x_k}M},
\]
cf. Example \ref{ex: T(Conf(M))}, therefore
\begin{align*}
 \Lambda^{k \cdot d}T_{(x_1,x_2, \cdots, x_k)}\OConf_k(M) & = \Lambda^{k \cdot d} (T_{x_1}M \oplus \cdots \oplus T_{x_k}M)\\ 
 & \cong \bigoplus_{p_1 + \cdots + p_k = k \cdot d} \Lambda^{p_1} T_{x_1}M \otimes \cdots \otimes \Lambda^{p_k} T_{x_k} M \\ 
 & = \Lambda^d T_{x_1}M \otimes \cdots \otimes \Lambda^d T_{x_k}M. 
\end{align*}
The last equality follows from the fact that $\Lambda^pT_{x} M = 0$ for all $p > d$ and for all $x \in M$. Or, in global writing:
\begin{align*}
 \Lambda^{k \cdot d}\bT\OConf_k(M) 
 & = \Lambda^d \pi_1^*TM \otimes \cdots \otimes \Lambda^d \pi_k^*TM. 
\end{align*}
Now, given $\nu_1\otimes \cdots\otimes \nu_k\in \Dens_{T_{x_1}M}\otimes \cdots\otimes \Dens_{T_{x_k}M}$, we set $\beta(\nu_1\otimes \cdots\otimes \nu_k):=\nu$, where 
\begin{align*}
\begin{array}{rcl}
 \nu: \ \Lambda^d T_{x_1}M \otimes \Lambda^d T_{x_2}M \otimes \cdots \otimes \Lambda^d T_{x_k}M & \longrightarrow & \R \\ 
 \omega_{x_1} \otimes \omega_{x_2} \otimes \cdots \otimes \omega_{x_k} & \mapsto& \nu_{x_1}(\omega_{x_1}) \nu_{x_2}(\omega_{x_2}) \cdots \nu_{x_k} (\omega_{x_k}).
 \end{array}
\end{align*} 
This $\beta$ is a fibrewise isomorphism. A local coordinate computation shows that $\beta$ varies smoothly with $(x_1,...,x_k)$, showing that $\beta$ is an isomorphism of vector bundles.

Finally, above a configuration $\underline{x}=\{x_1,...,x_k\}$, the composite map 
\begin{align}
 \delta\circ \gamma\circ \beta^{-1}\circ \alpha|_{(x_1,...,x_k)}^{-1}.\label{eq:longiso}
\end{align}
gives an isomorphism between the fibres $(\Dens_{\UConf(M)})_{\underline{x}}$ and $(S^{\boxtimes}\Dens_{M})_{\underline{x}}$. 
This map is  well-defined, because the isomorphisms $\beta$ and $\gamma$ are $S_k$-equivariant. Moreover the fibrewise isomorphisms assemble to a smooth bundle map, as can be seen by extending the isomorphism \eqref{eq:longiso} to a small open neighborhood of $\underline x$. Altogether we obtain the desired isomorphism 
\[
\Dens_{\UConf(M)} \cong S^{\boxtimes}\Dens_{M}.
\]
\end{proof}


\subsection{Cauchy-Hadamard $2$-algebra bundles}
\label{subsec:Cauchy-Hadamard 2-algebra bundles}

Following \cite[Definition 6.28]{Aguiar-Mahajan-2010}, let us call \emph{$2$-algebra bundle} over $\UConf(M)$ a double monoid in the $2$-monoidal category $(\VB(\UConf(M)),\boxtimes,\otimes)$, that is, a vector bundle $\bA\in \VB(\UConf(M))$ endowed with the following bundle maps 
\begin{align}\label{def:2-algebra bundle}
\begin{array}{rlcrl}
\text{$\otimes$-multiplication} & m_{\otimes}:\bA \otimes \bA \to \bA && \text{$\boxtimes$-multiplication} & m_{\boxtimes}:\bA \boxtimes \bA \to \bA \\ 
\text{$\otimes$-unit} & u_{\otimes}:\bI_{\otimes} \to \bA && \text{$\boxtimes$-unit} & u_{\boxtimes}:\bI_{\boxtimes} \to \bA 
\end{array}
\end{align}
such that $(\bA,m_{\otimes},u_{\otimes})$ and $(\bA,m_{\boxtimes},u_{\boxtimes})$ are unital associative algebras, in the usual sense, and the following diagrams commute: 
\begin{eqnarray}\label{eq:2-algebra}
& \hspace{.5cm}
\begin{tikzcd}
(\bA \otimes \bA) \boxtimes (\bA\otimes \bA) \arrow[rr, "\zeta"] \arrow[d, "m_{\otimes}\boxtimes m_{\otimes}"'] && 
(\bA \boxtimes \bA) \otimes (\bA \boxtimes \bA) \arrow[d, "m_{\boxtimes}\otimes m_{\boxtimes}"]
 \\ 
\bA \boxtimes \bA \arrow[dr, "m_{\boxtimes}"'] && 
\bA \otimes \bA \arrow[dl, "m_{\otimes}"] \\ 
& \bA & 
\end{tikzcd} 
\\ \nonumber
&
\begin{tikzcd}
\bI_{\otimes} \boxtimes \bI_{\otimes} \arrow[r, "u_{\otimes} \boxtimes u_{\otimes}"] \arrow[d, "\mu"] & 
\bA \boxtimes \bA \arrow[d, "m_{\boxtimes}"] \\ 
\bI_{\otimes} \arrow[r, "u_{\otimes}"'] & \bA
\end{tikzcd}
\qquad 
\begin{tikzcd}
\bI_{\boxtimes} \otimes \bI_{\boxtimes} \arrow[r, "u_{\boxtimes} \otimes u_{\boxtimes}"] & 
\bA \otimes \bA \arrow[d, "m_{\otimes}"]\\ 
\bI_{\boxtimes} \arrow[u, "\Delta"]\arrow[r, "u_{\boxtimes}"'] & \bA
\end{tikzcd} 
\qquad 
\begin{tikzcd}
\bI_{\boxtimes} \arrow[rr, "\nu"] \arrow[dr, "u_{\boxtimes}"'] && \bI_{\otimes} \arrow[dl, "u_{\otimes}"] \\
& \bA & 
\end{tikzcd} 
\end{eqnarray}
The $2$-algebra bundle $\bA$ is said to be \emph{commutative} if both multiplications are commutative. 
A \emph{morphism of $2$-algebra bundles} is a bundle map $\bA\to \bA'$ which respect to the multiplications and the units. 

\begin{example}
After Examples \ref{ex: H-unit = H-algebra} and \ref{ex: H-unit = C-algebra}, the unit bundle $\bI_{\otimes}$ is both a commutative $\otimes$-algebra and a commutative $\boxtimes$-algebra bundle over $\UConf(M)$. 
The above four diagrams commute, so finally $\bI_{\otimes}$ is a commutative $2$-algebra bundle over $\UConf(M)$. 
\end{example}

\begin{definition}
Let $V$ be a vector bundle on $M$. We call \emph{Cauchy-Hadamard bundle} on $V$ the Cauchy symmetric bundle $\bS^{\boxtimes} S^{\otimes}(V)$ of the usual symmetric bundle of $V$ on $M$. 
\end{definition}

The defining relations in eq. \eqref{eq:2-algebra} are easily verified, and give the following result. 

\begin{theorem}
The Cauchy-Hadamard bundle $\bS^{\boxtimes} S^{\otimes}(V)$ is a commutative $2$-algebra bundle on $\UConf(M)$, with the following operations: 
\begin{itemize}
\item $\boxtimes$-multiplication given by $\boxdot: \bS^{\boxtimes} S^{\otimes}(V) \boxtimes \bS^{\boxtimes} S^{\otimes}(V) \to \bS^{\boxtimes} S^{\otimes}(V)$,
\item $\boxtimes$-unit $\tilde{u}:\bI_{\boxtimes} \to \bS^{\boxtimes} S^{\otimes}(V)$,
\item $\otimes$-multiplication $\odot:\bS^{\boxtimes} S^{\otimes}(V) \otimes \bS^{\boxtimes} S^{\otimes}(V) \to \bS^{\boxtimes} S^{\otimes}(V)$ induced by the fibrewise free commutative product of the symmetric algebra bundle $S^{\otimes}(V)$: if for a configuration $\underline x=\{x_1,...,x_k\}$ an element $a_{\underline x}\in \bS^{\boxtimes} S^{\otimes}(V)$ is of the form $a_{\underline x}= a^1_{x_1}\boxdot \cdots \boxdot a^k_{x_k}$, with $a^i_{x_i}\in S^{\otimes}(V_{x_i})$, and similarly for another element $b_{\underline x}\in \bS^{\boxtimes} S^{\otimes}(V)$, then 
$$
a_{\underline x} \odot b_{\underline x} = (a^1_{x_1}\odot b^1_{x_1})\boxdot \cdots \boxdot (a^k_{x_k}\odot b^k_{x_k}), 
$$
\item $\otimes$-unit $\bI_{\otimes} \to \bS^{\boxtimes} S^{\otimes}(V)$ which identifies $1_{\emptyset}\in (I_{\otimes})_{\emptyset}$ with $1_{\emptyset}\in \bS^{\boxtimes 0} S^{\otimes}(V)_{\emptyset}$ and $1_{\underline x}\in (I_{\otimes})_{\underline x}$ with $1_{x_1}\odot \cdots \odot 1_{x_k}\in \bS^{\boxtimes k} S^{\otimes}(V)_{\underline x}$, where $1_{x_i}\in S^{\otimes}(V_{x_i})$. 
\end{itemize}
\end{theorem}


\section{Poisson bundles over configuration spaces}
\label{sec:4}

The goal of this section is to construct a Poisson algebra structure on the space of sections of the bundle $\bS^\boxtimes S^\otimes (V)\otimes \Dens_{\UConf(M)}$.
The Poisson bracket will be induced by a skew-symmetric kernel $k:V\boxtimes V\to I_\otimes$. Its extension to $\bS^\boxtimes S^\otimes (V)$ requires a Leibniz rule for both the $\boxtimes$-multiplication and the $\otimes$-multiplication. The resulting object is not just a Cauchy-Poisson algebra bundle, but also has an additional compatibility with the Hadamard structure. In Subsection \ref{subsec:Poisson Cauchy algebra} we adapt this construction to accommodate the tensor product with $\Dens_{\UConf(M)}$, and finally discuss the induced structure on the section space in Subsection \ref{subsec:sections}.


\subsection{Poisson $2$-algebra bundles}
\label{subsec:Poisson 2-algebra}

\begin{definition} \label{def:Poisson 2-algebra}
In the spirit of the previous definitions, we call \emph{Poisson $2$-algebra bundle} over $\UConf(M)$ a commutative $2$-algebra bundle $(\bP,m_{\otimes},u_{\otimes},m_{\boxtimes},u_{\boxtimes})$ endowed with a
\begin{align} \label{eq:Poisson bracket}
\text{Poisson bracket}\quad \{\ ,\ \}:\bP \boxtimes \bP \to \bP
\end{align}
satisfying the usual conditions with respect to both multiplications, namely: 
\begin{align*}
\text{antisymmetry: }\quad & \sum \{a_{\underline{x}},b_{\underline{y}}\} = - \sum \{b_{\underline{y}},a_{\underline{x}}\} \\ 
\text{Jacobi identity: }\quad & \sum \{\{a_{\underline{x}},b_{\underline{y}}\},c_{\underline{z}}\} +\sum \{\{b_{\underline{y}},c_{\underline{z}}\},a_{\underline{x}}\} +\sum \{\{c_{\underline{z}},a_{\underline{x}}\},b_{\underline{y}}\}=0 \\ 
\text{$m_{\boxtimes}$-Leibniz rule: }\quad & \sum \{a_{\underline{x}},m_{\boxtimes}(b_{\underline{y}},c_{\underline{z}})\} 
= \sum m_{\boxtimes}\left(\{a_{\underline{x}},b_{\underline{y}}\},c_{\underline{z}}\right) 
+ \sum m_{\boxtimes}\left(b_{\underline{y}},\{a_{\underline{x}},c_{\underline{z}}\}\right) \\ 
\text{$m_{\otimes}$-Leibniz rule: }\quad & \sum \{a_{\underline{x}},m_{\otimes}(b_{\underline{y}},c_{\underline{y}})\} 
= \sum m_{\otimes}\left(\{a_{\underline{x}},b_{\underline{y}}\},m_{\boxtimes} (1_{\underline{x}},c_{\underline{y}})\right) 
+ \sum m_{\otimes}\left(m_{\boxtimes}(1_{\underline{x}},b_{\underline{y}}),\{a_{\underline{x}},c_{\underline{y}}\}\right)
\end{align*}
for any $a,b,c\in \bP$ on the suitable (disjoint) configurations written in the formulas. 
To be precise, the antisymmetry condition lives in the fibre of $\bP\boxtimes \bP$ above the configuration $\underline{x}\sqcup \underline{y}$ and means that $\{\ ,\ \}\circ \tau_\boxtimes = -\{\ ,\ \}$ where $\tau_\boxtimes$ is the braiding of Lemma \ref{lem:monoidal}. 
The $m_\otimes$-Leibniz rule live in the fibre of $\bP \boxtimes (\bP\otimes \bP)$ above the same configuration $\underline{x}\sqcup \underline{y}$. 
The Jacobi identity and the $m_\boxtimes$-Leibniz rule live in the fibre of $\bP\boxtimes \bP\boxtimes \bP$ above the configuration $\underline{x}\sqcup \underline{y}\sqcup \underline{z}$. 
The identities are to be understood as holding for the sum over splitting in two or three (disjoint) configurations of any given configuration. 
\end{definition}

Our main example of a Poisson $2$-algebra is $S^{\boxtimes} S^{\otimes}(V)$, with $V\in \VB(M)$. This $2$-algebra has two gradings, one with respect to the $\boxtimes$-powers, and one with respect to the $\otimes$-powers. The total degree of a generic element is incremented by multiplying it by another element different from a unit, either with the commutative $\boxdot$-product or with the commutative $\odot$-product. The generators are then single vectors in $V$ over $1$-point configurations, i.e. points of $M$. 

In the sequel, let us adopt the short notation $V\boxtimes V$ for the vector bundle $i_*V \boxtimes i_*V$ over $\UConf_2(M)$. 

\begin{theorem} \label{thm:Cauchy-Hadamard-Poisson}
Let $V$ be a vector bundle on $M$. Any antisymmetric bundle map $k:V\boxtimes V\to \UConf_2(M)\times \K = \bI_{\otimes,2}$ over $\UConf_2(M)$ determines on the Cauchy-Hadamard $2$-algebra bundle $\bS^{\boxtimes} S^{\otimes}(V)$ the structure of a Poisson $2$-algebra bundle on $\UConf(M)$, in the following recursive way. 
For any (disjoint) configurations $\underline x$, $\underline y$ and $\underline z$: 
\begin{enumerate}
\item 
The bracket is zero whenever it is computed on a unit: 
\begin{align*}
\sum \{1_{\underline x},1_{\underline y}\} 
= \sum \{a_{\underline x},1_{\underline y}\} 
= \sum \{1_{\underline y},a_{\underline x}\} = 0
\end{align*}
for any $a_{\underline x}\in \bS^{\boxtimes} S^{\otimes}(V)_{\underline x}$. 

\item 
On the generators of $\bS^{\boxtimes} S^{\otimes}(V)$, that is, single vectors $u_x\in V_x$ and $v_y\in V_y$ above $1$-point configurations, the bracket is given by a function $k$:
\begin{align*}
 \{u_x,v_y\} = k(u_x,v_y) 1_x\boxdot 1_y \in I_x\boxdot I_y \subset \bS^{\boxtimes 2}S^\otimes(V)_{\{x,y\}}. 
\end{align*}

\item 
Finally, the bracket is extended to higher degree elements in $\bS^{\boxtimes} S^{\otimes}(V)$ by the Leibniz rule with respect to commutative Hadamard products $\odot$: 
\begin{align*}
\sum \{a_{\underline x}\odot b_{\underline x},c_{\underline y}\} 
& = \sum \{a_{\underline x},c_{\underline y}\} \odot (b_{\underline x}\boxdot 1_{\underline y}) 
+ \sum (a_{\underline x} \boxdot 1_{\underline y}) \odot \{b_{\underline x},c_{\underline y}\} \\ 
\sum \{a_{\underline x},b_{\underline y} \odot c_{\underline y}\} 
& = \sum \{a_{\underline x},b_{\underline y}\} \odot (1_{\underline x} \boxdot c_{\underline y}) 
+ \sum (1_{\underline x} \boxdot b_{\underline y}) \odot \{a_{\underline x},c_{\underline y}\} 
\end{align*}
and with respect to commutative Cauchy products $\boxdot$: 
\begin{align*}
\sum \{a_{\underline x} \boxdot b_{\underline y},c_{\underline z}\} 
& = \sum \{a_{\underline x},c_{\underline z}\} \boxdot b_{\underline y} 
+ \sum a_{\underline x} \boxdot \{b_{\underline y},c_{\underline z}\} \\ 
\sum \{a_{\underline x},b_{\underline y} \boxdot c_{\underline z}\} 
& = \sum \{a_{\underline x},b_{\underline y}\} \boxdot c_{\underline z} 
+ \sum b_{\underline y} \boxdot \{a_{\underline x},c_{\underline z}\}. 
\end{align*}
\end{enumerate}
\end{theorem}

We call the bundle map $k\in \Mor_{\UConf_2(M)}(V\boxtimes V,\bI_{\otimes,2})$ the \emph{kernel of the Poisson bracket} $\{\ ,\ \}$ on $\bS^{\boxtimes} S^{\otimes}(V)$. In the sequel, this particular Poisson bracket will be denoted $\{\ ,\ \}_k$. 

\begin{remark}
Note that in order to define the bracket on generators (point 2.), we have to guarantee that scalar values above $2$-point configurations $\{x,y\}$ belong to the algebra. Such scalar values are contained in $\bI_{\otimes,2}$ which is a subbundle of $\bS^\boxtimes S^\otimes(V)$ but not a subbundle of $\bS^\boxtimes(V)$. Indeed, $\bS^\boxtimes(V)$ only contains $\bI_\boxtimes$ which is a subbundle of $\bI_\otimes$ concentrated on $\UConf_0(M)$. This is the reason why we cannot define the Poisson bracket on the $\boxtimes$-algebra bundle $\bS^\boxtimes(V)$ and need to consider the $2$-algebra bundle $\bS^\boxtimes S^\otimes(V)$. 
\end{remark}

Before giving the proof, note that conditions 1. and 2. imply that all iterated brackets between single generators on different points of $M$ vanish. For instance
\begin{align*}
 \{\{u_x,v_y\},w_z\} &= k(u_x,v_y) \{1_{\{x,y\}},w_z\} = 0. 
\end{align*}

\begin{proof}
The Leibniz rules holds by definition of the bracket, once we checked that the recursive definition of the bracket is compatible with iterated products. We check the identity
\begin{align*}
 \sum \{(a_{\underline x} \odot b_{\underline x}) \odot c_{\underline x}, d_{\underline y}\} 
 &= \sum \{a_{\underline x} \odot (b_{\underline x} \odot c_{\underline x}), d_{\underline y}\}
\end{align*} 
by the following calculations:
\begin{align*}
 \{(a_{\underline x} \odot b_{\underline x}) \odot c_{\underline x}, d_{\underline y}\} &= \{a_{\underline x} \odot b_{\underline x}, d_{\underline y}\} \odot (c_{\underline x}\boxdot 1_{\underline y})
 + \big((a_{\underline x} \odot b_{\underline x}\big) \boxdot 1_{\underline y}) \odot \{ c_{\underline x},d_{\underline y}\}\\
 &= \{a_{\underline x}, d_{\underline y}\} \odot (b_{\underline x}\boxdot 1_{\underline y}) \odot (c_{\underline x}\boxdot 1_{\underline y}) \\
 &+ \{b_{\underline x}, d_{\underline y}\} \odot (a_{\underline x}\boxdot 1_{\underline y}) \odot (c_{\underline x}\boxdot 1_{\underline y}) 
+ \big((a_{\underline x} \odot b_{\underline x}\big) \boxdot 1_{\underline y}) \odot \{ c_{\underline x},d_{\underline y}\}\\
&= \{a_{\underline x}, d_{\underline y}\} \odot \big((b_{\underline x} \odot c_{\underline x})\boxdot 1_{\underline y}\big) \\
 &+ \{b_{\underline x}, d_{\underline y}\} \odot \big((a_{\underline x} \odot c_{\underline x})\boxdot 1_{\underline y}\big) 
+ \{ c_{\underline x},d_{\underline y}\} \odot \big((a_{\underline x} \odot b_{\underline x}) \boxdot 1_{\underline y}\big), 
\end{align*}
where we use the pentagon axiom of eq. \eqref{eq:2-algebra}.
By similar calculations we get 
\begin{align*}
 \sum \{(a_{\underline x} \boxdot a_{\underline y}) \odot (b_{\underline x} \boxdot b_{\underline y}), c_{\underline z}\} 
 &= \sum \{(a_{\underline x} \odot b_{\underline x}) \boxdot (a_{\underline y} \odot b_{\underline y}), c_{\underline z}\}.
\end{align*} 

We now prove the antisymmetry and the Jacobi identity by recursion.
The antisymmetry holds on the generators $u_x, v_y\in V$ because the function $k$ is antisymmetric, therefore
\begin{align*}
 \{u_x,v_y\} &= k(u_x,v_y) 1_x\boxdot 1_y 
 = -k(v_y,u_x) 1_y\boxdot 1_x 
 = \{v_y,u_x\}. 
\end{align*}
Suppose that it holds for all elements up to a given degree, let us check that it holds for higher degree elements: 
\begin{align*}
 \{a_{{\underline x}'}\boxdot b_{{\underline x}''},c_{\underline y}\} 
 &= \{a_{{\underline x}'},c_{\underline y}\} \boxdot b_{{\underline x}''} + a_{{\underline x}'} \boxdot \{b_{{\underline x}''},c_{\underline y}\} \\ 
 &= -\{c_{\underline y},a_{{\underline x}'}\} \boxdot b_{{\underline x}''} - a_{{\underline x}'} \boxdot \{c_{\underline y},b_{{\underline x}''}\} \\ 
 &= - \{c_{\underline y},a_{{\underline x}'}\boxdot b_{{\underline x}''}\} 
 \\ 
 \{a_{\underline x}\odot b_{\underline x},c_{\underline y}\} &= \{a_{\underline x},c_{\underline y}\} \boxdot (b_{\underline x}\boxdot 1_{\underline y}) + (a_{\underline x} \boxdot 1_{\underline y}) \odot \{b_{\underline x},c_{\underline y}\} \\ 
 &= -\{c_{\underline y},a_{\underline x}\} \odot (1_{\underline y} \boxdot b_{\underline x}) - (a_{\underline x} \boxdot 1_{\underline y}) \odot \{c_{\underline y},b_{\underline x}\} \\ 
 &= - \{c_{\underline y},a_{\underline x}\odot b_{\underline x}\}. 
\end{align*}

The Jacobi identity holds on generators, because in the Jacobian
\begin{align*}
 \mathrm{Jac}(u_x,v_y,w_z) &= \{u_x,\{v_y,w_z\}\} + \{v_y,\{w_z,u_x\}\} + \{w_z,\{u_x,v_y\}\}
\end{align*}
each of the three terms vanishes: $\{u_x,\{v_y,w_z\}\}=k(v_y,w_z) \{u_x,1_{\{y,z\}}\} = 0$. 
Suppose that the Jacobian vanishes for all elements up to a given degree, let us check that it vanishes for higher degree elements: in the Jacobian 
\begin{align*}
\mathrm{Jac}(a_{{\underline x}'}\boxdot b_{{\underline x}''},c_{\underline y},d_{\underline z}) &= 
\{a_{{\underline x}'}\boxdot b_{{\underline x}''},\{c_{\underline y},d_{\underline z}\}\} + \{c_{\underline y},\{d_{\underline z},a_{{\underline x}'}\boxdot b_{{\underline x}''}\}\} + \{d_{\underline z},\{a_{{\underline x}'}\boxdot b_{{\underline x}''},c_{\underline y}\}\}, 
\end{align*}
the first term is 
\begin{align*}
 & \{a_{{\underline x}'},\{c_{\underline y},d_{\underline z}\}\} \boxdot b_{{\underline x}''} + a_{{\underline x}'}\boxdot \{b_{{\underline x}''},\{c_{\underline y},d_{\underline z}\}\}, 
\end{align*}
the second term is 
\begin{align*}
 & \hspace{-1cm} \{c_{\underline y},\{d_{\underline z},a_{{\underline x}'}\} \boxdot b_{{\underline x}''} \} 
 + \{c_{\underline y},\{d_{\underline z},b_{{\underline x}''}\} \boxdot a_{{\underline x}'} \} \\ 
 & = \{c_{\underline y},\{d_{\underline z},a_{{\underline x}'}\} \} \boxdot b_{{\underline x}''} 
 + \{d_{\underline z},a_{{\underline x}'}\} \boxdot \{c_{\underline y},b_{{\underline x}''} \} \\ 
 & \hspace{1cm} + \{c_{\underline y},\{d_{\underline z},b_{{\underline x}''}\} \} \boxdot a_{{\underline x}'} 
 + \{d_{\underline z},b_{{\underline x}''}\} \boxdot \{c_{\underline y},a_{{\underline x}'} \} , 
\end{align*}
and the third term is 
\begin{align*}
 & \hspace{-1cm} \{d_{\underline z},\{a_{{\underline x}'},c_{\underline y}\}\boxdot b_{{\underline x}''}\} 
 + \{d_{\underline z},a_{{\underline x}'}\boxdot \{b_{{\underline x}''},c_{\underline y}\}\} \\ 
 & = \{d_{\underline z},\{a_{{\underline x}'},c_{\underline y}\}\} \boxdot b_{{\underline x}''} 
 + \{a_{{\underline x}'},c_{\underline y}\}\boxdot \{d_{\underline z},b_{{\underline x}''}\} \\ 
 & \hspace{1cm} + \{d_{\underline z},a_{{\underline x}'}\} \boxdot \{b_{{\underline x}''},c_{\underline y}\}
 + a_{{\underline x}'}\boxdot \{d_{\underline z},\{b_{{\underline x}''},c_{\underline y}\}\} . 
\end{align*}
Therefore
\begin{align*}
 \mathrm{Jac}(a_{{\underline x}'}\boxdot b_{{\underline x}''},c_{\underline y},d_{\underline z}) &= 
 \mathrm{Jac}(a_{{\underline x}'},c_{\underline y},d_{\underline z}) \boxdot b_{{\underline x}''} 
 + a_{{\underline x}'}\boxdot \mathrm{Jac}(b_{{\underline x}''},c_{\underline y},d_{\underline z}) \\ 
 & \hspace{1cm} + \{d_{\underline z},a_{{\underline x}'}\} \boxdot \Big[\{c_{\underline y},b_{{\underline x}''}\}+\{b_{{\underline x}''},c_{\underline y}\}\Big] \\ 
 & \hspace{1cm} + \{d_{\underline z},b_{{\underline x}''}\} \boxdot \Big[\{c_{\underline y},a_{{\underline x}'}\} + \{a_{{\underline x}'},c_{\underline y}\}\Big] \\ 
 & = 0.
\end{align*}
Similarly, in the Jacobian
\begin{align*}
\mathrm{Jac}(a_{\underline x}\odot b_{\underline x},c_{\underline y},d_{\underline z}) &= 
\{a_{\underline x}\odot b_{\underline x},\{c_{\underline y},d_{\underline z}\}\} + \{c_{\underline y},\{d_{\underline z},a_{\underline x}\odot b_{\underline x}\}\} + \{d_{\underline z},\{a_{\underline x}\odot b_{\underline x},c_{\underline y}\}\}, 
\end{align*}
the first term is 
\begin{align*}
 & \{a_{\underline x},\{c_{\underline y},d_{\underline z}\}\} \odot (b_{\underline x}\boxdot 1_{{\underline y} \sqcup {\underline z}}) + (a_{\underline x}\boxdot 1_{{\underline y} \sqcup {\underline z}}) \odot \{b_{\underline x},\{c_{\underline y},d_{\underline z}\}\}, 
\end{align*}
the second term is 
\begin{align*}
 & \hspace{-1cm} \{c_{\underline y},\{d_{\underline z},a_{\underline x}\} \odot (b_{\underline x}\boxdot 1_{\underline z})\} 
 + \{c_{\underline y},(a_{\underline x} \boxdot 1_{\underline z}) \odot \{d_{\underline z},b_{\underline x}\}\} \\ 
 & = \{c_{\underline y},\{d_{\underline z},a_{\underline x}\}\} \odot (b_{\underline x} \boxdot 1_{{\underline y} \sqcup {\underline z}})
 + \{c_{\underline y},b_{\underline x}\boxdot 1_{\underline z}\} \odot (1_{\underline y} \boxdot \{d_{\underline z},a_{\underline x}\}) \\ 
 & \hspace{1cm} + \{c_{\underline y},a_{\underline x} \boxdot 1_{\underline z}\} \odot (1_{\underline y} \boxdot \{d_{\underline z},b_{\underline x}\})
 + (a_{\underline x} \boxdot 1_{{\underline y} \sqcup {\underline z}}) \odot \{c_{\underline y},\{d_{\underline z},b_{\underline x}\} \} \\ 
 & = \{c_{\underline y},\{d_{\underline z},a_{\underline x}\}\} \odot (b_{\underline x} \boxdot 1_{{\underline y} \sqcup {\underline z}}) 
 + (\{c_{\underline y},b_{\underline x}\} \boxdot 1_{\underline z}) \odot (1_{\underline y} \boxdot \{d_{\underline z},a_{\underline x}\}) \\ 
 & \hspace{1cm} + (\{c_{\underline y},a_{\underline x}\} \boxdot 1_{\underline z}) \odot (1_{\underline y} \boxdot \{d_{\underline z},b_{\underline x}\}) 
 + (a_{\underline x} \boxdot 1_{{\underline y} \sqcup {\underline z}}) \odot \{c_{\underline y},\{d_{\underline z},b_{\underline x}\} \} , 
\end{align*}
because $\{c_{\underline y},b_{\underline x} \boxdot 1_{\underline z}\} = \{c_{\underline y},b_{\underline x}\} \boxdot 1_z + b_{\underline x} \boxdot \{c_{\underline y},1_{\underline z}\}$ with $\{c_{\underline y},1_{\underline z}\}=0$, 
and the third term is 
\begin{align*}
 & \hspace{-1cm} \{d_{\underline z},\{a_{\underline x},c_{\underline y}\}\odot (b_{\underline x}\boxdot 1_{\underline y})\} 
 + \{d_{\underline z},(a_{\underline x}\boxdot 1_{\underline y}) \odot \{b_{\underline x},c_{\underline y}\}\} \\ 
 & = \{d_{\underline z},\{a_{\underline x},c_{\underline y}\}\} \odot (b_{\underline x}\boxdot 1_{{\underline y} \sqcup {\underline z}}) 
 + (1_{\underline z} \boxdot \{a_{\underline x},c_{\underline y}\}) \odot \{d_{\underline z},b_{\underline x}\boxdot 1_{\underline y}\} \\
 & \hspace{1cm} + \{d_{\underline z},a_{\underline x} \boxdot 1_{\underline y}\} \odot (\{b_{\underline x},c_{\underline y}\} \boxdot 1_{\underline z}) 
 + (a_{\underline x}\boxdot 1_{{\underline y} \sqcup {\underline z}}) \odot \{d_{\underline z},\{b_{\underline x},c_{\underline y}\}\} \\ 
 & = \{d_{\underline z},\{a_{\underline x},c_{\underline y}\}\} \odot (b_{\underline x}\boxdot 1_{{\underline y} \sqcup {\underline z}}) 
 + (1_{\underline z} \boxdot \{a_{\underline x},c_{\underline y}\}) \odot (\{d_{\underline z},b_{\underline x}\} \boxdot 1_{\underline y}) \\
& \hspace{1cm} + (\{d_{\underline z},a_{\underline x}\} \boxdot 1_{\underline y}) \odot (\{b_{\underline x},c_{\underline y}\} \boxdot 1_{\underline z}) 
 + (a_{\underline x}\boxdot 1_{{\underline y} \sqcup {\underline z}}) \odot \{d_{\underline z},\{b_{\underline x},c_{\underline y}\}\}. 
\end{align*}
Therefore
\begin{align*}
 \mathrm{Jac}(a_{\underline x} \odot b_{\underline x},c_{\underline y},d_{\underline z}) &= 
 \mathrm{Jac}(a_{\underline x},c_{\underline y},d_{\underline z}) \odot (b_{\underline x} \boxdot 1_{{\underline y} \sqcup {\underline z}})
 + (a_{\underline x}\boxdot 1_{{\underline y} \sqcup {\underline z}}) \odot \mathrm{Jac}(b_{\underline x},c_{\underline y},d_{\underline z}) \\ 
 & \hspace{1cm} + (1_{\underline y} \boxdot \{d_{\underline z},a_{\underline x}\}) \odot 
 \left(\Big[\{c_{\underline y},b_{\underline x}\} + \{b_{\underline x},c_{\underline y}\}\Big] \boxdot 1_{\underline z}\right) \\
 & \hspace{1cm} + (1_{\underline y} \boxdot \{d_{\underline z},b_{\underline x}\}) \odot 
 \left(\Big[\{c_{\underline y},a_{\underline x}\} + \{a_{\underline x},c_{\underline y}\}\Big] \boxdot 1_{\underline z}\right) \\ 
 & = 0.
\end{align*}
\end{proof}

\begin{example}
Let $V\to M$ be a vector bundle. For any $a\in V_x$ and any $n>0$, denote by $a^n = a^{\otimes n} = a\odot \cdots \odot a$ the symmetric $\otimes$-power of $a$ of degree $n$, which is an element of $S^{\otimes n}(V_x)$. Also set $a^0 = 1_x \in S^{\otimes 0}(V)=\K$. Then, for $a\in V_x$ and $b\in V_y$ we have 
\begin{align*}
 \{ a^n,b^m \} &= nm\, k(a,b)\, a^{n-1}\boxdot b^{m-1}, 
\end{align*}
which is an element of $S^{\otimes n-1}(V_x) \boxdot S^{\otimes m-1}(V_y) \subset \bS^{\boxtimes 2}(S^\otimes(V))_{\{x,y\}}$. 

Now consider $k=k'+k''$ elements $a_i\in V_{x_i}$ and $b_j\in V_{y_j}$ for $i=1,...,k'$ and $j=1,...,k''$, and the corresponding symmetric $\boxtimes$- and $\otimes$-power of multi-degree $(n_1,...,n_{k'})$ and $(m_1,...,m_{k''})$. Then we have 
\begin{align*}
 \{ a_1^{n_1}\boxdot \cdots \boxdot a_{k'}^{n_{k'}}, b_1^{m_1}\boxdot \cdots \boxdot b_{k''}^{m_{k''}} \} &= \\ 
 & \hspace{-3cm}
 \underset{j=1,...,k''}{\sum_{i=1,...,k'}} n_i m_j\, k(a_i,b_j)\, a_1^{n_1}\boxdot \cdots \boxdot a_i^{n_i-1}\boxdot \cdots \boxdot a_{k'}^{n_{k'}} \boxdot b_1^{m_1}\boxdot \cdots \boxdot b_j^{m_j-1}\boxdot \cdots \boxdot b_{k''}^{m_{k''}}. 
\end{align*}
This is an element of 
\begin{align*}
 \bS^{\boxtimes k'} S^{\otimes}(V)_{\underline{x}}\boxdot \bS^{\boxtimes k''} S^{\otimes}(V)_{\underline{y}} \subset \bS^{\boxtimes k} S^{\otimes}(V)_{\underline{x}\sqcup \underline{y}}
\end{align*}
with $\underline{x}=\{x_1,...,x_{k'}\}$ and $\underline{y}=\{y_1,...,y_{k''}\}$. 

This example shows that the two Leibniz rules (with respect to the $\boxtimes$- and to the $\otimes$-structures) are a proper model of a biderivation, as should be a Poisson bracket. This will be clearer on the induced Poisson bracket of sections. 
\end{example}


\subsection{Poisson-Cauchy algebra bundles}
\label{subsec:Poisson Cauchy algebra}

\begin{definition}
A \emph{Poisson $\boxtimes$-algebra bundle} over $\UConf(M)$ is a commutative $\boxtimes$-algebra bundle $(\bP,m_{\boxtimes},u_{\boxtimes})$ endowed with a Poisson bracket bundle map $\{\ ,\ \}$ as in Definition \ref{def:Poisson 2-algebra} satisfying only the $m_{\boxtimes}$-Leibniz rule. 
\end{definition}

A Poisson $2$-algebra bundle is clearly a Poisson $\boxtimes$-algebra bundle if we forget the $\otimes$-structure. 
In order to describe the Poisson structure of multilocal polynomial functionals in field theory (which will be the main object of study in \cite{Frabetti-Kravchenko-Ryvkin-2024}), we are interested in the Poisson $\boxtimes$-algebra induced by a $2$-point function which will be a \emph{generalized} function, i.e. a distribution: this will require an underlying $2$-algebra structure which is relevant only for the Poisson structure, and also an extension of the values of the bundle by densities on the base manifold. This extension can be done with respect to any $\boxtimes$-algebra bundle, by means of the Hadamard tensor product.

\begin{proposition} \label{prop:P otimes A is Poisson}
Let $\bP$ be a Poisson $\boxtimes$-algebra bundle on $\UConf(M)$ with multiplication $\bullet_P$ and bracket $\{\ ,\ \}$, and let $\bA$ be a commutative $\boxtimes$-algebra bundle on $\UConf(M)$ with multiplication $\bullet_A$. Then the Hadamard tensor product $\bP\otimes \bA$ is a Poisson $\boxtimes$-algebra bundle with bracket 
\begin{align*}
 \{p_{\underline x} \otimes a_{\underline x},q_{\underline y}\otimes b_{\underline y}\} 
 := \{p_{\underline x},q_{\underline y}\} \otimes a_{\underline x} \bullet_A b_{\underline y}
\end{align*}
for any (disjoint) configurations $\underline x$, $\underline y$, any $p_{\underline x},q_{\underline y}\in \bP$ and any $a_{\underline x},b_{\underline y}\in \bA$. 
\end{proposition}

\begin{proof}
The Hadamard product of the two commutative $\boxtimes$-algebras is easily verified to be a commutative $\boxtimes$-algebra, with multiplication 
$$
(p_{\underline x} \otimes a_{\underline x}) \bullet (q_{\underline y} \otimes b_{\underline y}) = (p_{\underline x} \bullet_P q_{\underline y}) \otimes (a_{\underline x} \bullet_A b_{\underline y}). 
$$
Then the result on the antisymmetry and the Jacobi identity of the Poisson bracket is well known and corresponds to the extension of a Lie algebra structure to the tensor product of a Lie algebra by a commutative algebra. 
It remains to verify that the Leibniz rule holds with respect to the product $\bullet$: 
\begin{align*}
 & \hspace{-1cm}
 \{(p'_{{\underline x}'} \otimes a'_{{\underline x}'}) \bullet (p''_{{\underline x}''} \otimes a''_{{\underline x}''}),q_{\underline y} \otimes b_{\underline y}\} \\ 
 &= \{(p'_{{\underline x}'} \bullet_P p''_{{\underline x}''}) \otimes (a'_{{\underline x}'} \bullet_A a''_{{\underline x}''}),q_{\underline y} \otimes b_{\underline y}\} \\
 &= \{p'_{{\underline x}'} \bullet_P p''_{{\underline x}''},q_{\underline y}\} \otimes (a'_{{\underline x}'} \bullet_A a''_{{\underline x}''} \bullet_A b_{\underline y}) \\ 
 &= \Big(\{p'_{{\underline x}'},q_{\underline y}\} \bullet_P p''_{{\underline x}''}\Big) \otimes (a'_{{\underline x}'} \bullet_A b_{\underline y} \bullet_A a''_{{\underline x}''}) 
 + \Big(p'_{{\underline x}'} \bullet_P \{p''_{{\underline x}''},q_{\underline y}\}\Big) \otimes (a'_{{\underline x}'} \bullet_A a''_{{\underline x}''} \bullet_A b_{\underline y}) \\
 & = \Big(\{p'_{{\underline x}'},q_{\underline y}\} \otimes (a'_{{\underline x}'} \bullet_A b_{\underline y}) \Big) \bullet (p''_{{\underline x}''}\otimes a''_{{\underline x}''}) 
 + (p'_{{\underline x}'} \otimes a'_{{\underline x}'}) \bullet \Big(\{p''_{{\underline x}''},q_{\underline y}\} \otimes (a''_{{\underline x}''} \bullet_A b_{\underline y}) \Big) \\ 
 & = \{p'_{{\underline x}'}\otimes a'_{{\underline x}'},q_{\underline y}\otimes b_{\underline y}\} \bullet (p''_{{\underline x}''}\otimes a''_{{\underline x}''}) 
 + (p'_{{\underline x}'} \otimes a'_{{\underline x}'}) \bullet \{p''_{{\underline x}''} \otimes a''_{{\underline x}''},q_{\underline y} \otimes b_{\underline y}\}. 
\end{align*}
\end{proof}

\begin{corollary} \label{cor:Poisson-Cauchy algebra bundle}
Let $V\to M$ and $W\to M$ be two vector bundles on $M$. 
Then, for any choice of an antisymmetric bundle map $k:V\boxtimes V \to \UConf_2(M)\times \K$, the vector bundle 
\begin{align*}
 \bS^{\boxtimes}\big(S^{\otimes}(V) \otimes W\big) \cong \bS^{\boxtimes} S^{\otimes}(V) \otimes \bS^{\boxtimes}(W) \to \UConf(M)
\end{align*}
is a Poisson $\boxtimes$-algebra bundle, with structure maps given on any (disjoint) configurations $\underline{x},\underline{y}\in \UConf(M)$, any elements $a_{\underline x},b_{\underline y}\in S^{\boxtimes} S^{\otimes}(V)$ and any elements $\alpha_{\underline x},\beta_{\underline y}\in S^{\boxtimes}(W)$, by 
\begin{align*}
\text{commutative multiplication: }\quad & 
(a_{\underline x} \otimes \alpha_{\underline x}) \boxdot (b_{\underline y} \otimes \beta_{\underline y}) 
= (a_{\underline x} \boxdot b_{\underline y}) \otimes (\alpha_{\underline x} \boxdot \beta_{\underline y}) \\ 
\text{Poisson bracket: }\quad & 
\{a_{\underline x} \otimes \alpha_{\underline x},b_{\underline y} \otimes \beta_{\underline y}\} 
= \{a_{\underline x},b_{\underline y}\}_k \otimes (\alpha_{\underline x} \boxdot \beta_{\underline y}), 
\end{align*}
where $\{\ ,\ \}_k$ is the Poisson bracket on $\bS^{\boxtimes}S^{\otimes}(V)$ given in Theorem \ref{thm:Cauchy-Hadamard-Poisson}.
\end{corollary}

When $W$ is the density bundle $\Dens_M$ on $M$ (see Section \ref{subsec:density}), we call $\bS^{\boxtimes} \big(S^{\otimes}(V)\otimes \Dens_M\big)\cong \bS^{\boxtimes} S^\otimes(V) \otimes \Dens_{\UConf(M)}$ the \emph{Poisson-Cauchy algebra bundle} on $V$. This bundle will be essential for applying our constructions to field theory in the forthcoming article \cite{Frabetti-Kravchenko-Ryvkin-2024}, cf. also Section \ref{subsec:sections} below.


\subsection{Sections of Poisson algebra bundles}
\label{subsec:sections}

For physical applications, the main interest of Poisson algebra bundles over $\UConf(M)$ is the induced structure on their spaces of sections.\\

A \emph{smooth section} of a vector bundle $\pi:\bV\to \UConf(M)$ is a smooth map $\phi:\UConf(M)\to \bV$ (where smoothness is seen component-wise) such that $\pi\circ\phi=\id_{\UConf(M)}$. 
Since $\UConf(M)=\bigsqcup_k \UConf_k(M)$, the space of smooth sections of $\bV$ is
\[
\Gamma(\UConf(M),\bV)= \prod_k \Gamma(\UConf_k(M),\bV_k).
\]
As usual, smooth sections of $\bV$ can also be seen as bundle maps from the trivial line bundle $\bI_\otimes = \UConf(M)\times \K$ to $\bV$ and more generally the morphisms $\Mor(\bV,\bW)$ are in bijection with $\Gamma(\UConf(M),\Hom(\bV,\bW))$. 
Sections of the trivial line bundle give \emph{smooth functions} on $\UConf(M)$, and their set
\[
C^\infty(\UConf(M))=\Gamma(\UConf(M),\bI_\otimes) 
\] 
forms a unital commutative algebra with the usual multiplication $(f\, g)(\underline{x}) = f(\underline{x}) g(\underline{x})$, where $f$ and $g$ are functions and $\underline{x}\in \UConf(M)$. 
More generally, the section space $\Gamma(\UConf(M),\bA)$ of any $\otimes$-algebra bundle $(\bA,m_\otimes)$ inherits the structure of an algebra with multiplication $\phi \cdot\psi = m_\otimes\circ (\phi\otimes\psi)$ and unit $1(\underline{x})=1$ for any $\underline{x}\in \UConf(M)$. 
A similar result holds for $\boxtimes$-algebra bundles. 

\begin{lemma} \label{lem:sections algebra} 
If $\bA$ is a $\boxtimes$-algebra bundle over $\UConf(M)$ with multiplication $\bullet_\bA$, then its space of smooth sections $\Gamma(\UConf(M),\bA)$ is an algebra with multiplication 
\begin{align} \label{eq:sections multiplication}
 (\phi_1\bullet \phi_2)(\underline{x}) &= \sum_{\underline{x}= \underline{x}'\sqcup \underline{x}''} \phi_1(\underline{x}') \bullet_{\bA} \phi_2(\underline{x}'') 
\end{align}
for any sections $\phi_1$, $\phi_2$ and any configuration $\underline{x}$ and its unit is given by the section
\begin{align} \label{eq:sections unit}
 1(\underline{x}) &= \begin{cases}
 1 & \text{if $\underline{x}=\emptyset$} \\
 0 & \text{otherwise}. 
 \end{cases}
\end{align}
If $\bA$ is commutative, so is $\Gamma(\UConf(M),\bA)$. 
\end{lemma}

\begin{proof}
As pointed out above, sections of $\bA$ over $\UConf(M)$ can be equivalently described as bundle maps $\UConf(M)\times \K=\bI_\otimes \to \bA$ over $\UConf(M)$. From Theorem \ref{thm: VB(Conf(M)) is 2-monoidal}, we know that $\bI_\otimes$ is a $\boxtimes$-algebra. Here, we rather need that it is also naturally a $\boxtimes$-coalgebra (that is, a $\boxtimes$-comonoid in the sense of \cite{Aguiar-Mahajan-2010}, see Remark \ref{rem:Cauchy Sigma}) with comultiplication $\Delta:\bI_\otimes\to \bI_\otimes\boxtimes \bI_\otimes$ given by 
\begin{align*}
\Delta(1_{\underline x})=\sum_{\underline x=\underline x'\sqcup\underline x''}1_{\underline x'}\otimes 1_{\underline x''}. 
\end{align*}
Then, if $\bullet_\bA$ denotes the multiplication $m:\bA\boxtimes \bA\to \bA$ in the $\boxtimes$-algebra bundle $\bA$, two sections $\phi_1,\phi_2: \bI_\otimes \to \bA$ can be multiplied by means of the convolution product $\mathfrak{m}(\phi_1,\phi_2)=m\circ(\phi_1\boxtimes \phi_2)\circ\Delta$, i.e. as the composition
\begin{align*}
\bI_\otimes\overset{\Delta}\longrightarrow \bI_\otimes\boxtimes \bI_\otimes \overset{\phi_1\boxtimes \phi_2}\longrightarrow \bA\boxtimes \bA \overset{m}\longrightarrow \bA. 
\end{align*}
This gives a well-defined associative operation, since $\bullet_A$ is associative and $\Delta$ is coassociative. It is easy to check that $\mathfrak{m}$ gives the multiplication $\bullet$ of \eqref{eq:sections multiplication}. 
Since $\Delta$ is co-commutative, a commutative bundle multiplication $m$ induces a commutative multiplication $\mathfrak{m}$ on sections.

To define the unit section in $\Gamma(\UConf(M),\bA)$, instead of considering the Cauchy unit bundle $\bI_\boxtimes$ as a natural sub-bundle of $\bI_\otimes$, with inclusion $\nu$ as in Theorem \ref{thm: VB(Conf(M)) is 2-monoidal}, let us consider the natural (dual) projection $\nu^*:\bI_\otimes \to \bI_\boxtimes$ which is the identity on the zero-point configuration and the zero map on all other configurations. 
Then, the unit section $1$ is the section corresponding to the bundle morphism $u\circ \nu^*:\bI_\otimes \to \bI_\boxtimes \to \bA$, where $u$ is the unit of $\bA$. 
Explicitly, this gives the above section \eqref{eq:sections unit}.
\end{proof}

We note that in the $\otimes$-case, the section spaces are $C^\infty(\UConf(M))$-algebras, while in the $\boxtimes$-case we obtain $\mathbb K$-algebras. Note also that the above procedure works when we have two (compatible) algebraic structures at once:

\begin{lemma} \label{lem:sections Poisson}
If $\bP\to \UConf(M)$ is a Poisson $\boxtimes$-algebra over $\UConf(M)$ with bracket $\{\ ,\ \}_{\bP}$, then its space of smooth sections $\Gamma(\UConf(M),\bP)$ is a Poisson algebra, with Poisson bracket 
\begin{align} \label{eq:sections bracket}
 \{\phi_1,\phi_2\}(\underline{x}) &= \sum_{\underline{x}= \underline{x}'\sqcup \underline{x}''} \{\phi_1(\underline{x}'),\phi_2(\underline{x}'')\}_\bP 
\end{align}
for any sections $\phi_1$, $\phi_2$ and any configuration $\underline{x}$. 
\end{lemma}

\begin{proof}
By Lemma \ref{lem:sections algebra}, we already know that $\Gamma(\UConf(M),\bP)$ is a commutative algebra. Let us use the same trick as for proving Lemma \ref{lem:sections algebra}: the (non-associative) Poisson bracket on $\bP$ is a bundle map $\{\ ,\ \}_{\bP}:\bP\boxtimes \bP\to \bP$. Then take two sections $\phi_1,\phi_2$ of $\bP$, see them as two bundle maps $\bI_\otimes \to \bP$ and define their bracket as the bundle map $\{\phi_1,\phi_2\}=\{\ ,\ \}_{\bP}\circ(\phi_1\boxtimes \phi_2)\circ \Delta$, i.e. as the composition 
\begin{align*}
\bI_\otimes\overset{\Delta}\longrightarrow \bI_\otimes\boxtimes \bI_\otimes \overset{\phi_1\boxtimes \phi_2}\longrightarrow \bP\boxtimes \bP \overset{\{\ ,\ \}_\bP}\longrightarrow \bP. 
\end{align*}
It is easy to check that this gives the above formula \eqref{eq:sections bracket}.

The Poisson identities for $\{\ ,\ \}$ then follow from those of $\{\ ,\ \}_\bP$ and from two properties of the splitting of configurations: the fact that the sum over splittings $\underline{x}=\underline{x}'\sqcup \underline{x}''$ is the same as the sum over splittings $\underline{x}=\underline{x}''\sqcup \underline{x}'$, and the fact that the sum over repeated splittings is (co)associative, i.e. the sums over splittings $\underline{x}=(\underline{x}'_{(1)}\sqcup \underline{x}'_{(2)})\sqcup \underline{x}''$ and $\underline{x}=\underline{x}'\sqcup (\underline{x}''_{(1)}\sqcup \underline{x}''_{(2)})$ coincide, and can be denoted $\underline{x}=\underline{x}'\sqcup \underline{x}''\sqcup \underline{x}'''$. 
The antisymmetry is then clear. For the Jacobi and the Leibniz rule (with respect to $\bullet$), it suffices to write the explicit result for any three sections $\phi_1$, $\phi_2$, $\phi_3$ and any configuration $\underline{x}$: 
\begin{align*}
 \{\phi_1,\{\phi_2,\phi_3\}\}(\underline{x}) &= \sum_{\underline{x}=\underline{x}'\sqcup \underline{x}''\sqcup \underline{x}'''} \{\phi_1(\underline{x}'),\{\phi_2(\underline{x}''),\phi_3(\underline{x}''')\}_\bP \}_\bP \\ 
 \{\phi_1,\phi_2\bullet \phi_3\}(\underline{x}) & = \sum_{\underline{x}=\underline{x}'\sqcup \underline{x}''\sqcup \underline{x}'''} \{\phi_1(\underline{x}'),\phi_2(\underline{x}'')\bullet_\bP \phi_3(\underline{x}''') \}_\bP. 
\end{align*}
The Jacobi and the Leibniz rules then hold on sections because by assumption they hold on each fibre. 
\end{proof}

We now present an example of a Poisson algebra constructed with the above techniques, which plays a key role in field theory. 

\begin{example} 
Let $E\to M$ be a vector bundle underlying some field theory. Smooth functionals on $\Gamma(M,E)$ give classical off-shell observables, where smooth refers to functional derivatives \cite[p.11]{Gelfand-Fomin-1963}. The space of functionals $C^\infty(\Gamma(M,E))$ is endowed with the usual commutative pointwise multiplication and contains a dense subalgebra which is also closed under the Peierls Poisson bracket determined by a given Lagrangian density $\L:JE\to \Dens_M$ \cite{Peierls-1952} \cite{Rejzner-2016}. 
Here and in the sequel, the jet bundle $JE$ (of a given order) is necessary to account for terms involving the jet prolongation $j\varphi$ of fields, i.e. its partial derivatives. 

We are interested in recovering this Poisson algebra as the space of (distributional) sections of a suitable Poisson algebra bundle $\bP_\L(E)$. We start with the natural interpretation of sections of $(JE)^*\otimes \Dens_M$ with compact support as smooth functionals on $\Gamma(M,E)$ acting by integration, that is, the linear map 
\begin{align*}
 F:\Gamma_c(M,(JE)^*\otimes \Dens_M)\to C^{\infty}(\Gamma(M,E)),\ \phi \otimes \nu \mapsto F_{\phi\otimes \nu} 
\end{align*}
with smooth function $F_{\phi\otimes \nu}$ of the field $\varphi\in \Gamma(M,E)$ given by 
\begin{align*}
 F_{\phi\otimes \nu}(\varphi) &= \int_M \langle \phi(x),j\varphi(x)\rangle \nu(x).
\end{align*}
By considering the usual symmetric algebra $S_\R(\Gamma_c(M,(JE)^*\otimes \Dens_M))$ of real vector spaces (which enjoys the universal property of free commutative algebras), the above integration extends to an algebra map 
$$
S_\R(\Gamma_c(M,(JE)^*\otimes \Dens_M))\to C^\infty(\Gamma(M,E)). 
$$
However the symmetric algebra of sections is not a section space. The settings of Cauchy algebra bundles over $\UConf(M)$ allows us to see $i_*((JE)^*\otimes \Dens_M)$ as a sub-bundle of $\bS^\boxtimes((JE)^*\otimes \Dens_M)$ and therefore gives a natural linear map $\Gamma_c(M,(JE)^*\otimes \Dens_M)\to \Gamma(\UConf(M),\bS^\boxtimes((JE)^*\otimes \Dens_M))$. 
By Lemma \ref{lem:sections algebra}, the latter space is an algebra, therefore we also have an algebra map 
$$
S_\R(\Gamma_c(M,(JE)^*\otimes \Dens_M))\to \Gamma(\UConf(M),\bS^\boxtimes((JE)^*\otimes \Dens_M)). 
$$
The section space $\Gamma(\UConf(M),\bS^\boxtimes((JE)^*\otimes \Dens_M))$ is the first natural candidate to represent observables in field theory, but it is not closed under the Peierls bracket.

As shown in \cite{Peierls-1952}, the Peierls bracket is completely fixed by the causal propagator $\Delta^C_\L$ determined by the Lagrangian $\L$, which is a distribution on $M\times M$ with values in antisymmetric tensors inside $JE\exttens JE$ \cite[Eq. (4.4)]{Rejzner-2016}. 
In \cite{Frabetti-Kravchenko-Ryvkin-2024}, we show that $\Delta^C_\L$ determines a distribution on $\UConf_2(M)$ with values in the bundle of antisymmetric tensors $\Lambda^{c,\boxtimes 2}(JE) \subset JE \boxtimes JE$ (cf. Remark \ref{rem:Cauchy Sigma}), which is isomorphic to the bundle $\Hom(\Lambda^{\boxtimes 2}(JE)^*,\bI_{\otimes,2})$. 
Because of the standard bijection $\Gamma(\UConf_2(M),\Hom(\Lambda^{\boxtimes 2}(JE)^*,\bI_{\otimes,2})) \cong \Mor(\Lambda^{\boxtimes 2}(JE)^*,\bI_{\otimes,2})$, we see that a smooth regular version of the causal propagator is exactly the kernel of a Poisson bracket on the $\boxtimes$-algebra bundle $\bS^\boxtimes S^\otimes (JE)^* \otimes \Dens_{\UConf(M)}$. 
Notice that the Hadamard symmetric bundle $S^\otimes (JE)^*$ is necessary to get scalar values of the propagator supported at $2$-point configurations $\{x,y\}$. If we drop the Hadamard symmetric bundle, the only possible scalar values in the Cauchy symmetric bundle $\bS^\boxtimes (JE)^*$ would be supported at the vacuum. 

The Poisson $\boxtimes$-algebra bundle $\bS^\boxtimes S^\otimes (JE)^* \otimes \Dens_{\UConf(M)}$ is then a good candidate to represent polynomial observables, and we denote it by $\bP_\L(E)$. However, one needs to extend the whole Poisson structure to distributional sections of $\bP_\L(E)$ for Peierls bracket to be  well-defined. 

Moreover, not all (distributional) sections of $\bP_\L(E)$ yield well-defined functions on the space of fields $\Gamma(M,E)$, which carry no restrictions on their support. In particular, requiring compact supports on the sections or distributions of $\bP_\L(E)$ does not solve the problem, since compactly supported sections are not closed under $\boxtimes$-multiplication, cf. Remark \ref{rem:noncompact}. 
The appropriate space of distributional sections of $\bP_\L(E)$, together with the explicit integration map to $C^\infty(\Gamma(M,E))$, will be described in detail in the second part \cite{Frabetti-Kravchenko-Ryvkin-2024} of this work. 
\end{example}


\section{Conclusion and outlook}
\label{sec:5}

In this article we have settled the algebraic / geometric basis suitable to describe a Poisson algebra bundle $\bP_\L(E)$ whose (distributional) sections represent multilocal observables of classical relativistic fields $\varphi:M\to E$ ruled by a Lagrangian $\L$. 

The key to achieve this is to consider vector bundles over the (unordered) configuration space $\UConf(M)$, where points in $M$ can be switched, together with a symmetric $2$-monoidal structure given by the usual (Hadamard) tensor product $\otimes$ and a new (Cauchy) tensor product $\boxtimes$ which provides a symmetrized version of the usual external tensor product of vector bundles on $M$ (cf. Theorem \ref{thm:Cauchy vs external}). 

The algebra structure of the bundle $\bP_\L(E)$ is that of a free algebra with respect to the Cauchy tensor product, and reflects the fact that multilocal observables are polynomial functions on some generators. It turns out that $\bP_L(E)=\bS^\boxtimes\big(S^\otimes(JE)^*\otimes \Dens_M\big)$, where the bundle $S^\otimes(JE)^*\otimes \Dens_M$ represents local observables. 

The Poisson structure of the bundle $\bP_L(E)$ is induced by a bilinear and antisymmetric kernel $k:(JE)^*\boxtimes (JE)^*\to \K$ designed on the causal propagator laying at the core of Peierls bracket in relativistic field theory. The generators on which this kernel applies, namely the bundle $(JE)^*$, are not those of the algebra structure: the extension of the Poisson bracket from its kernel on $(JE)^*$ to the bundle $S^\otimes(JE)^*\otimes \Dens_M$ called for the new notion of a Poisson $2$-algebra bundle, where a new adapted $\otimes$-Leibniz rule is required in addition to the standard $\boxtimes$-Leibniz one. 

The bundle setup to describe multilocal observables is achieved. The goal of the companion paper \cite{Frabetti-Kravchenko-Ryvkin-2024} is to describe the suitable space of distributional sections of $\bP_\L(E)$ which admits a (well-defined) linear map 
$$
F:\D'_{int}(\UConf(M),\bP_\L(E))\to C^\infty(\Gamma(M,E)).
$$ 
On the dense subset of regular distributions $\phi \otimes \nu$ with values in $\bP_\L(E)$, the observable $F_{\phi\otimes \nu}$ assigns to a field $\varphi:M\to E$ the value 
\begin{align*}
 F_{\phi\otimes \nu}(\varphi) = \int_{\UConf(M)} \langle \phi(\underline{x}), e(j\varphi)(\underline{x}) \rangle \nu(\underline{x}), 
\end{align*}
where $e(j\varphi)$ is an \emph{ad-hoc} exponential section of the sub-bundle $\bSigma^\boxtimes \hat{\Sigma}^\otimes(JE)$ of symmetric Cauchy and Hadamard tensors (cf. Remarks \ref{rem:Hadamard Sigma} and \ref{rem:Cauchy Sigma}), induced by the jet prolongation of $\varphi$. 
This is only possible because $\Dens_{\UConf(M)} \cong \bS^\boxtimes(\Dens_M)$ (cf. Theorem \ref{thm:Dens_Conf(M)=S(Dens_M)}), and therefore 
$$
\bP_\L(E) = \bS^\boxtimes\big(S^\otimes (JE)^*\otimes \Dens_M\big) \cong \bS^\boxtimes S^\otimes (JE)^* \otimes \Dens_{\UConf(M)} \cong \big(\bSigma^\boxtimes \hat{\Sigma}^\otimes(JE)\big)^\vee. 
$$

\addcontentsline{toc}{section}{References}
\bibliographystyle{plain} 
\bibliography{literature}

\end{document}